\documentclass[envcountsame,runningheads]{llncs}

\usepackage[dvipsnames,svgnames]{xcolor}

\usepackage{todonotes}

\usepackage{amsmath,amssymb}
\usepackage{mathtools}

\usepackage{xcolor,pgf,tikz}
\usetikzlibrary{arrows,automata,shapes,decorations}

\newif\ifdraft\drafttrue

\ifdraft

\newcommand\mycomment[3]{\todo[inline,size=\scriptsize,backgroundcolor=#1,caption={}]{%
    \begin{minipage}{\textwidth-4pt}#3 - \textbf{#2}\end{minipage}}}
\newcommand\nf[1]{\mycomment{Pink}{Nath}{#1}}
\newcommand\el[1]{\mycomment{SpringGreen}{Engel}{#1}}
\newcommand\po[1]{\mycomment{Goldenrod}{Pierre}{#1}}
\newcommand\ap[1]{\mycomment{PaleTurquoise}{Amaury}{#1}}
\newcommand\jo[1]{\mycomment{DarkGray}{Joël}{#1}}

\else

\newcommand\nf[1]{}
\newcommand\el[1]{}
\newcommand\po[1]{}
\newcommand\ap[1]{}
\newcommand\jo[1]{}

\fi

\newcommand{\rpm}{\raisebox{.2ex}{$\scriptstyle\pm$}}

\newcommand{\set}[1]{\left\{ #1 \right\}}

\newcommand{\C}{\mathbb{C}}
\newcommand{\U}{\mathbb{U}}
\newcommand{\R}{\mathbb{R}}

\newcommand{\Q}{\mathbb{Q}}
\newcommand{\A}{\mathbb{A}}
\newcommand{\Z}{\mathbb{Z}}
\newcommand{\N}{\mathbb{N}}
\newcommand{\nats}{\mathbb{N}}

\newcommand{\I}{\mathcal{I}}
\newcommand{\J}{\mathcal{J}}

\renewcommand{\H}{\mathcal{H}}

\newcommand{\Jord}[2]{\mathcal J_{ #1 } (#2)}

\renewcommand{\Re}[1]{\text{Re}\left( #1 \right)}
\renewcommand{\Im}[1]{\text{Im}\left( #1 \right)}

\renewcommand{\O}{\mathcal{O}}
\renewcommand{\P}{\mathcal{P}}

\newcommand{\interior}[1]{{\kern0pt#1}^{\mathrm{o}}}
\newcommand{\closure}[1]{\overline{ #1 }}
\newcommand{\convex}[1]{\text{Conv}\left( #1 \right)}
\newcommand{\Section}[1]{\text{Section}\left( #1 \right)}
\newcommand{\compl}[1]{ #1 ^{\mathsf{c} }}

\newcommand{\diag}{\text{Diag}}

\newcommand{\last}{\text{last}}
\newcommand{\lasttwo}{\text{last-two}}

\newcommand{\F}{\mathcal{Q}}

\def\eps{\varepsilon}

\begin{document}

\title{On the Monniaux Problem\\ in Abstract Interpretation%
\thanks{Nathana{\"e}l Fijalkow, Pierre Ohlmann, and Amaury Pouly were supported by the Agence Nationale de la Recherche through
the project Codys (ANR-18-CE40-0007). Jo\"el Ouaknine was supported by ERC grant AVS-ISS (648701) and by DFG grant
389792660 as part of TRR 248 (see
https://perspicuous-computing.science). James Worrell was supported by EPSRC Fellowship EP/N008197/1.}}


 \author{Nathana{\"e}l Fijalkow\inst{1} \and
 Engel Lefaucheux\inst{2} \and
   Pierre Ohlmann\inst{3} \and
   Jo\"el Ouaknine\inst{2,4} \and
   Amaury Pouly\inst{5} \and
 James Worrell\inst{4}
 }

\authorrunning{N. Fijalkow et al.}
 \institute{CNRS, LaBRI, France and The Alan Turing Institute, UK \and
 Max Planck Institute for Software Systems, Saarland Informatics
 Campus, Germany \and
 IRIF, Universit\'e Paris 7, France \and
 Department of Computer Science, Oxford University, UK\and
 CNRS, IRIF, Universit\'e Paris Diderot, France}

\bibliographystyle{plain}

\maketitle

\begin{abstract}
    The Monniaux Problem in abstract interpretation asks,\linebreak roughly speaking, whether the
    following question is decidable: given a program
    $P$, a safety (\emph{e.g.}, non-reachability) specification $\varphi$, and an abstract
    domain of invariants $\mathcal{D}$, 
    does there exist an inductive invariant $\I$ in $\mathcal{D}$
    guaranteeing that program $P$ meets its specification $\varphi$.
    The Monniaux Problem is of course parameterised by the classes of
    programs and invariant domains that one considers.

    In this paper, we show that the Monniaux Problem is
    undecidable for unguarded affine programs and semilinear invariants
    (unions of polyhedra). Moreover, we show that decidability is
    recovered in the important special case of simple linear loops.
\end{abstract}


\section{Introduction}

Invariants are one of the most fundamental and useful notions in the
quantitative sciences, appearing in a wide range of contexts, from
gauge theory, dynamical systems, and control theory in physics,
mathematics, and engineering to program verification, static analysis,
abstract interpretation, and programming language semantics (among
others) in computer science. In spite of decades of scientific work
and progress, automated invariant synthesis remains a topic of active
research, especially in the fields of program
analysis and abstract interpretation, and plays a central role in
methods and tools seeking to establish correctness properties of
computer programs; see, \emph{e.g.},~\cite{KCBR18}, and particularly Sec.~8
therein.

The focus of the present paper is the \textbf{\emph{Monniaux Problem}}
on the decidability of the existence of separating invariants, which
was formulated by David Monniaux in~\cite{Mon17b,Mon17} and also raised by
him in a series of personal
communications with various members of the theoretical computer
science community over the past five years or so. There are in fact a
multitude of versions of the Monniaux Problem---indeed, it would be
more appropriate to speak of a \emph{class} of problems rather than a single
question---but at a high level the formulation below is one of the
most general:
\begin{quote}
  Consider a program $P$ operating over some numerical domain (such as
  the integers or rationals), and assume that $P$ has an underlying
  finite control-flow graph over the set of nodes
  $Q = \{q_1, \ldots, q_r\}$. Let us assume that $P$ makes use of $d$
  numerical variables, and each transition $q \overset{t}{\longrightarrow} q'$
  comprises a function
  $f_t: \mathbb{R}^d \rightarrow \mathbb{R}^d$ as well as a
  guard $g_{t} \subseteq \mathbb{R}^d$. Let
  $x,y \in \mathbb{Q}^d$ be two points in the ambient space. By way
  of intuition and motivation, we are interested in the reachability problem as to
  whether, starting in location $q_1$ with variables having valuation
  $x$, it is possible to reach location $q_r$ with variables having
  valuation $y$, by following the available transitions and under the
  obvious interpretation of the various functions and
  guards. Unfortunately, in most settings this problem is well-known
  to be undecidable.

    Let $\mathcal{D} \subseteq 2^{\mathbb{R}^d}$ be an `abstract domain' for $P$,
    \emph{i.e.}, a collection of subsets of
    $\mathbb{R}^d$. For example, $\mathcal{D}$ could be the
      collection of all convex polyhedra in $\mathbb{R}^d$, or the
      collection of all closed semialgebraic sets in $\mathbb{R}^d$, etc.

    The Monniaux Problem can now be formulated as a decision question:
    is it possible to adorn each control location $q$ with an element
    $\I_q \in \mathcal{D}$ such that:
    \begin{enumerate}
    \item $x \in \I_{q_1}$;
      \item The collection of $\I_q$'s forms an \emph{inductive
          invariant}:
        for each transition $q \overset{t}{\longrightarrow} q'$, we
        have that $f_t(\I_q \cap g_t) \subseteq \I_{q'}$; and
        \item $y \notin \I_{q_r}$.
        \end{enumerate}

        We call such a collection $\{\I_q : q \in Q\}$ a
        \emph{separating inductive invariant} for program
        $P$. (Clearly, the existence of a separating inductive invariant 
        constitutes a proof of non-reachability for $P$ with the given
        $x$ and $y$.)

        Associated with this decision problem, in positive instances
        one is also potentially interested in the synthesis problem, \emph{i.e.}, the
        matter of algorithmically producing a suitable separating
        invariant $\{\I_q : q \in Q\}$.\footnote{In the remainder of
          this paper, the term `invariant' shall always refer to the
          inductive kind.}
\end{quote}

The Monniaux Problem is therefore parameterised by a number of items,
key of which are (i)~the abstract domain $\mathcal{D}$ under
consideration, and (ii)~the kind of functions and guards allowed in
transitions.

Our main interest in this paper lies in the \emph{decidability} of the
existence of separating invariants for various instances of the
Monniaux Problem. We give below a cursory cross-sectional survey of
existing work and results in this direction.

Arguably the earliest positive result in this area is due to Karr, who
showed that strongest affine invariants (conjunctions of
affine equalities) for affine programs (no guards, and all transition
functions are given by affine expressions) could be
computed algorithmically~\cite{Karr76}. Note that the ability
  to synthesise strongest (\emph{i.e.}, smallest with respect to set
  inclusion) invariants immediately entails the decidability of the
  Monniaux Problem instance, since the existence of \emph{some}
  separating invariant
  is clearly equivalent to whether the
  \emph{strongest} invariant is separating.
M\"{u}ller-Olm and Seidl later extended this work on
affine programs to include the computation of strongest polynomial
invariants of fixed degree~\cite{Muller-OlmS04}, and
a randomised algorithm for discovering affine relations
was proposed by Gulwani and Necula [16].
More recently, Hrushovski \emph{et al.}\
showed how to compute a basis for \emph{all} polynomial
relations at every location of a given affine
program~\cite{HOPW18}.

The approaches described above all compute invariants consisting
exclusively of conjunctions of \emph{equality} relations. By contrast,
an early and highly influential paper by Cousot and Halbwachs
considers the domain of convex closed polyhedra~\cite{CH78}, for
programs having polynomial transition functions and guards. Whilst no
decidability results appear in that paper, much further work was
devoted to the development of restricted polyhedral domains for which
theoretical guarantees could be obtained, leading (among others) to
the \emph{octagon domain} of Min\'e~\cite{Mine01}, the
\emph{octahedron domain} of Claris\'o and Cortadella~\cite{CC04}, and
the \emph{template polyhedra} of Sankaranarayanan \emph{et
  al.}~\cite{SSM05}. In fact, as observed by Monniaux~\cite{Mon17}, if
one considers a domain of convex polyhedra having a \emph{uniformly
  bounded} number of faces (therefore subsuming in particular the domains just
described), then for any class of programs with polynomial transition
relations and guards, the existence of separating invariants becomes
decidable, as the problem can equivalently be phrased in the first-order theory
of the reals.

One of the central motivating questions for the Monniaux Problem is
whether one can always compute separating invariants for the full
domain of polyhedra. Unfortunately, on this matter very little is
known at present. In recent work, Monniaux showed undecidability for
the domain of convex polyhedra and the class of programs having affine
transition functions and polynomial guards~\cite{Mon17}. One of the
main results of the present paper is to show undecidability for the
domain of \emph{semilinear sets}\footnote{A semilinear set consists of
  a finite union of polyhedra, or equivalently is defined as the
  solution set of a Boolean combination of linear inequalities.}
and the class of affine programs (without
any guards)---in fact, affine programs with only a single control
location and two transitions:

\begin{theorem}\label{th:main-undec}
    Let $A, B \in \mathbb{Q}^{d \times d}$ be two rational square matrices of
    dimension $d$, and let $x,y \in \mathbb{Q}^d$ be two points in
    $\mathbb{Q}^d$. Then the existence of a semilinear set $\I \subseteq
    \mathbb{R}^d$ having the following properties:
    \begin{enumerate}
    \item $x \in \I$;
    \item $A\I \subseteq \I$ and $B\I \subseteq \I$; and
    \item $y \notin \I$
    \end{enumerate}
    is an undecidable problem.
\end{theorem}

\begin{remark}
 It is worth pointing out that the theorem remains valid even for
  sufficiently large fixed $d$ (our proof shows undecidability for
  $d=336$, but this value could undoubtedly be improved). If moreover
  one requires $\I$ to be topologically closed, one can lower $d$ to
  having fixed value $27$ (which again is unlikely to be optimal).
  Finally, an examination of the proof reveals that the theorem also
  holds for the domain of semialgebraic sets, and in fact for any
  domain of o-minimal sets in the sense of~\cite{ACO018}. The proof
  also carries through whether one considers the domain of semilinear sets
  having rational, algebraic, or real coordinates.
\end{remark}

Although the above is a negative (undecidability) result, it should be
viewed in a positive light; as Monniaux writes in~\cite{Mon17},
\emph{``We started this work hoping to vindicate forty years of
  research on heuristics by showing that the existence of polyhedral
  inductive separating invariants in a system with transitions in
  linear arithmetic (integer or rational) is undecidable.''}
Theorem~\ref{th:main-undec} shows that, at least as regards non-convex
invariants, the development and use of heuristics is indeed vindicated and
will continue to remain essential. Related questions
of \emph{completeness} of given abstraction scheme have also been
examined by Giaccobazzi \emph{et al.} in~\cite{GRS00,GLR15}.

It is important to note that our undecidability result requires at
least \emph{two} transitions. In fact, much research work has been
expended on the class of simple \emph{affine} loops, \emph{i.e.}, one-location programs
equipped with a single self-transition. In terms of invariants,
Fijalkow \emph{et al.} establish in~\cite{FOOPW17,FOOPW19} the decidability of
the existence of \emph{semialgebraic} separating invariants, and
specifically state the question of the existence of separating
\emph{semilinear} invariants as an open problem. Almagor \emph{et al.} extend
this line of work in~\cite{ACO018} to more complex targets (in lieu of
the point $y$) and richer classes of invariants. The second main
result of the present paper is to settle the open question
of~\cite{FOOPW17,FOOPW19} in the affirmative:

\begin{theorem}\label{one-matrix-theorem}
    Let $A \in \mathbb{Q}^{d \times d}$ be a rational square matrix of
    dimension $d$, and let $x,y \in \mathbb{Q}^d$ be two points in
    $\mathbb{Q}^d$. It is decidable whether there exists a closed semilinear set 
    $\I \subseteq \mathbb{R}^d$ having algebraic coordinates such that:
    \begin{enumerate}
    \item $x \in \I$;
    \item $A\I \subseteq \I$; and
    \item $y \notin \I$.
    \end{enumerate}
\end{theorem}

\begin{remark}
The proof shows that, in fixed dimension $d$, the decision procedure
runs in polynomial time. It is worth noting that one also has decidability
if $A$, $x$, and $y$ are taken to have
real-algebraic (rather than rational) entries.
\end{remark}

Let us conclude this section by briefly commenting on the important
issue of \emph{convexity}. At its inception, abstract interpretation
had a marked preference for domains of \emph{convex} invariants, of
which the interval domain, the octagon domain, and of course the
domain of convex polyhedra are prime examples. Convexity confers
several distinct advantages, including simplicity of representation,
algorithmic tractability and scalability, ease of implementation, and
better termination heuristics (such as the use of widening). The
central drawback of convexity, on the other hand, is its poor
expressive power. This has been noted time and again: \emph{``convex
  polyhedra 
[\ldots]\
are insufficient for expressing certain
  invariants, and what is often needed is a disjunction of convex
  polyhedra.''}~\cite{BM18}; \emph{``the ability to express non-convex
  properties is sometimes required in order to achieve a precise
  analysis of some numerical properties''}~\cite{GIBMG12}. Abstract
interpretation can accommodate non-convexity either by introducing
disjunctions (see, \emph{e.g.}, \cite{BM18} and references therein), or via
the development of special-purpose domains of non-convex invariants
such as \emph{donut domains}~\cite{GIBMG12}. The
technology, data structures, algorithms, and heuristics supporting the
use of disjunctions in the leading abstract-interpretation tool
\textsc{Astr\'ee} are presented in great detail in~\cite{CCFMMR09}. In
the world of software verification, where predicate
abstraction is the dominant paradigm, disjunctions---and hence
non-convexity---are nowadays native features of the landscape.

It is important to note that the two main results presented in this
paper, Theorems~\ref{th:main-undec} and \ref{one-matrix-theorem}, have
only been proven for families of invariants that are not necessarily
convex. The Monniaux Problem restricted to families of \emph{convex}
invariants remains open and challenging.

\section{Preliminaries}\label{sec:prel}

\subsection{Complex and algebraic numbers}

The set of complex numbers is $\C$, and for a complex number $z$ its modulus is $|z|$, its real part is $\Re{z}$ and its imaginary part is $\Im{z}$.

Let $\C^*$ denote the set of non-zero complex numbers.
We write $S^1$ for the complex unit circle, \textit{i.e.} the set of complex numbers of modulus $1$.
We let $\U$ denote the set of roots of unity, \textit{i.e.} complex numbers $z\in S^1$ such that $z^n = 1$ for some $n \in \N$.

When working in $\C^d$, the norm of a vector $z$ is $||z||$, defined as the maximum of the moduli of each complex number $z_i$ for $i$ in $\set{1,\ldots,d}$.
For $\varepsilon > 0$ and $z$ in $\C^d$, we write $B(z,\varepsilon)$ for the open ball centered in $z$ of radius $\varepsilon$.
The topological closure of a set $\I \subseteq \C^d$ is $\closure{\I}$, its interior $\interior{\I}$, and its frontier $\partial{\I}$,
defined as $\closure{\I} \cap \closure{\C^d \setminus \I}$.

We will mostly work in the field $\A \subseteq \C$ of algebraic numbers, that is, roots of polynomials with coefficients in $\Z$.
It is possible to represent and manipulate algebraic numbers effectively, by storing their minimal polynomial and
a sufficiently precise numerical approximation. An excellent reference
in computational algebraic number theory is~\cite{Cohen}. 
All standard algebraic operations such as sums, products, root-finding of polynomials, or computing Jordan normal forms of
matrices with algebraic entries can be performed effectively.

\subsection{Semilinear sets}

\vskip1em
We now define semilinear sets in $\C^d$, by identifying $\C^d$ with $\R^{2d}$.
A set $\I \subseteq \R^{2d}$ is semilinear if it is the set of real solutions of some finite Boolean combination of linear inequalities with algebraic coefficients.
We give an equivalent definition now using half-spaces and polyhedra.
A half-space $\H$ is a subset of $\C^d$ of the form 
\[
\H = \set{z \in \C^d \mid \sum_{i=1}^d \Re{z \overline u} \succ a},
\]
for some $u$ in $\A^d$, $a$ in $\A \cap \R$ and $\mbox{$\succ$}\in\set{\geq,>}$.
A polyhedron is a finite intersection of half-spaces, and a semilinear set a finite union of polyhedra.

\vskip1em
We recall some well known facts about semilinear sets which will be useful for our purposes.

\begin{lemma}[Projections of Semilinear Sets]
\label{lem:projection}
Let $\I$ be a semilinear set in $\C^{d+d'}$.
Then the projection of $\I$ on the first $d$ coordinates, defined by
\[
\Pi(\I,d) = \set{z \in \C^d \mid \exists t \in \C^{d'}, (z,t) \in \I}
\]
is a semilinear set.
\end{lemma}

\begin{lemma}[Sections of Semilinear Sets]
\label{lem:section}
Let $\I$ be a semilinear set in $\C^{d+d'}$ and $t$ in $\C^{d'}$.
Then the section of $\I$ along $t$, defined by 
\[
\Section{\I,t} = \set{z \in \C^d \mid (z,t) \in \I},
\] 
is a semilinear set.

Furthermore, there exists a bound $B$ in $\R$ such that for all $t$ in $\C^{d'}$ of norm at most $1$,
if $\Section{\I,t}$ is non-empty, then it contains some $z$ in $\C^d$ of norm at most $B$.
\end{lemma}

For the reader's intuitions, note that the last part of this lemma does not hold for more complicated sets.
For instance, consider the hyperbola defined by $\I = \set{(x,y) \in \R^2 \mid x y = 1}$.
Choosing a small $x$ forces to choose a large $y$, hence there exist no bound $B$ as stated in the lemma for~$\I$.

The dimension of a set $X$ of $\R^d$ is the minimal $k$ in $\N$ such that 
$X$ is included in a finite union of affine subspaces of dimension at most $k$.

\begin{lemma}[Dimension of Semilinear Sets]
\label{lem:dimension}
Let $\I$ be a semilinear set in $\R^d$. 
If $\interior{\I} = \emptyset$, then $\I$ has dimension at most $d - 1$.
\end{lemma}

\section{Main results overview}\label{sec:results}
We are interested in instances of the Monniaux Problem in which there
are no guards, all transitions are affine (or equivalently linear,
since affine transitions can be made linear by increasing the
dimension of the ambient space by 1), and invariants are
semilinear. This gives rise to the \emph{semilinear invariant
  problem}, where an instance is given by a set of square matrices
$A_1,\dots,A_k \in \A^{d \times d}$ and two points $x, y \in \A^d$. A
semilinear set $\I \subseteq \C^d$ is a \emph{separating invariant} if
\begin{enumerate}
	\item $x \in \I$,
	\item $A_i \I \subseteq \I$ for all $i\leq k$,
	\item $y \notin \I$.
\end{enumerate}
The semilinear invariant problem asks whether such an invariant exists.

We need to introduce some terminology.
The triple $((A_i)_{i\leq k},x,y)$ is a 
\emph{reach} instance if there exists a matrix 
$M$ belonging to the semigroup generated by $(A_i)_{i\leq k}$
such that $M x = y$, and otherwise it is
a \emph{non-reach} instance.
Clearly a separating invariant can only exist for non-reach instances.
An instance for $k=1$ is called
an \emph{Orbit instance}.

\subsection{Undecidability for several matrices}

Our first result is the undecidability of the semilinear invariant problem. We
start by showing it is undecidable in fixed dimension, with a fixed number
of matrices and requiring that the invariant be closed. We defer the proofs
until Section~\ref{sec:undec}.

\begin{theorem}\label{th:closed-monniaux}
    The semilinear invariant problem is undecidable for 9 matrices of dimension 3
    and closed invariants.
\end{theorem}

In establishing the above, we used many matrices of small dimension.
One could instead use only two matrices, but increasing the dimension to 27.

\begin{theorem}\label{th:closed-monniaux-bis}
    The semilinear invariant problem is undecidable for 2 matrices of dimension 27 and closed
    invariants.
\end{theorem}

In the above results, it can happen that the target belongs to the closure of the set of reachable points.
We now show that we can ignore those ``non-robust'' systems and
maintain undecidability.

\begin{theorem}\label{th:closed-monniaux-robust}
    The semilinear invariant problem is undecidable for 
    ``robust'' instances, \emph{i.e.} instances in which
    the target point does not belong to the closure of the set of reachable points.
\end{theorem}

The proof of the above result does not require that the invariants be
closed. We can therefore establish Theorem~\ref{th:main-undec} by
making use of the same
construction as in the proof of Theorem~\ref{th:closed-monniaux-bis} to
encode all the matrices of Theorem~\ref{th:closed-monniaux-robust} into
only two distinct matrices.



\subsection{Decidability for simple linear loops}

In this section, we are only concerned with Orbit instances.
Since it is possible to decide (in polynomial time) whether an Orbit instance is reach or non-reach~\cite{KL80,KL86},
we can always assume that we are given a non-reach instance. All
decidability results are only concerned with \emph{closed invariants},
this is crucial in several proofs.

\begin{theorem}\label{thm:main}
    There is an algorithm that decides whether an Orbit instance admits a
    closed semilinear invariant. Furthermore, it runs in polynomial time assuming
    the dimension d is fixed.
\end{theorem}

We now comment a few instructive examples to illustrate the different cases that arise. The proof of Theorem~\ref{thm:main} is postponed to Section~\ref{sec:dec}.

\begin{example}\label{ex:1}
    Consider the Orbit instance $\ell = (A,x,y)$ in dimension $2$ where
    \[
    A = \frac 1 2 \begin{bmatrix} 1 & -2 \\ 2 & 1\end{bmatrix},
    \]
    $x=(1,0)$ and $y=(3,3)$. The orbit is depicted on Figure~\ref{fig:ex1}.
    Here, $A$ is a counterclockwise rotation around the origin with an expanding scaling factor.
    A suitable semilinear invariant can be constructed by taking the complement of the convex hull
    of a large enough number of points of the orbit, and adding the missing points. In this example,
    we can take
    \[
    \I = \{x,Ax\} \cup \compl{\convex{\{A^n x, n \leq 8\}}}.
    \]
\end{example}

\begin{figure}[ht]
    \centering
    \includegraphics[scale=0.5]{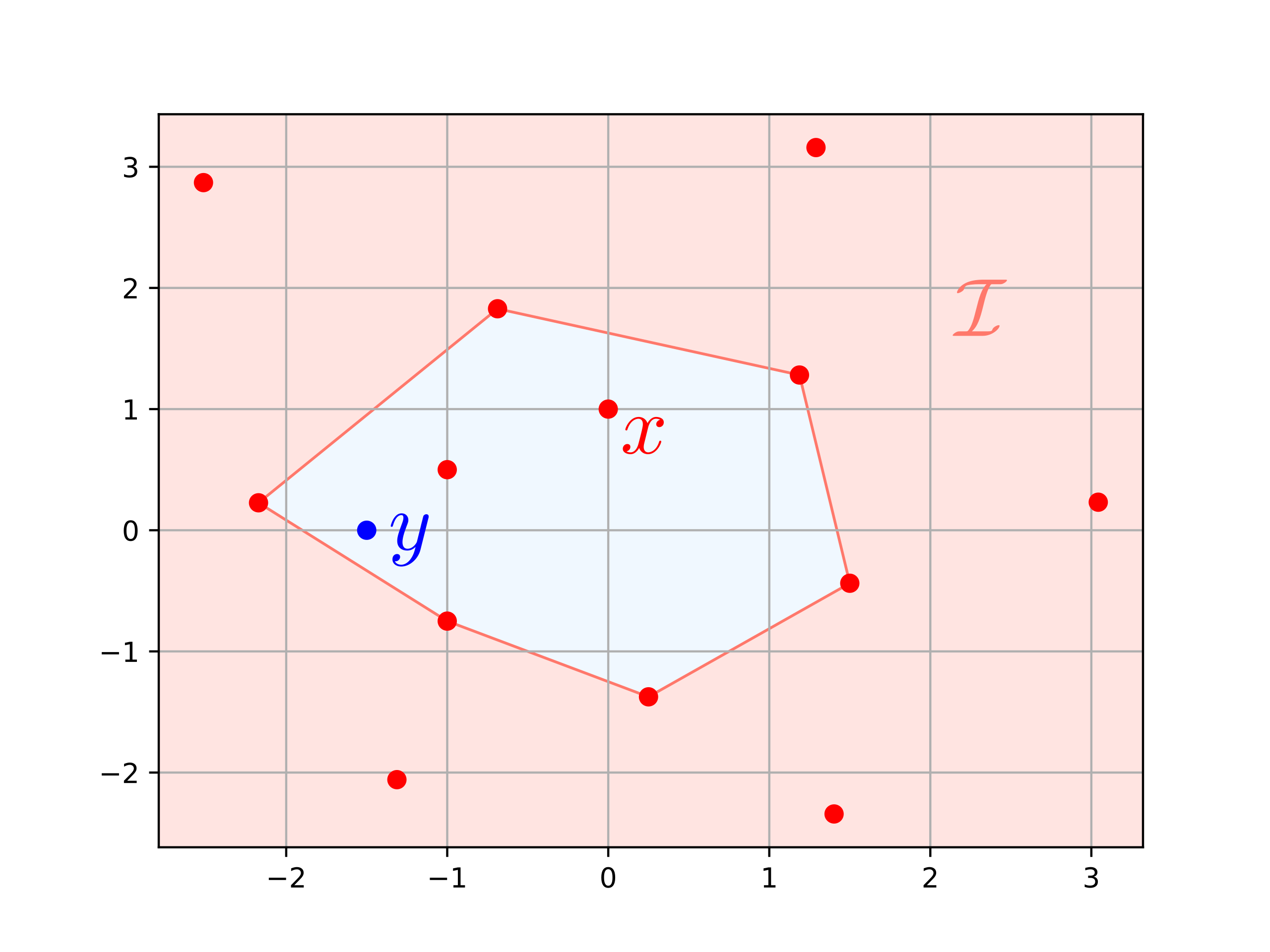}
    \caption{An invariant for example~\ref{ex:1}.\label{fig:ex1}}
\end{figure}

Constructing an invariant of this form will often be possible, for instance when $A$ has an eigenvalue of modulus $>1$.
A similar (yet more involved) construction gives the same result when $A$ has an eigenvalue of modulus $<1$.
The case in which all eigenvalues have modulus 1 is more involved.
Broadly speaking, invariant properties in such cases are often better described by
sets involving equations or inequalities of higher degree~\cite{FOOPW17},
which is why interesting semilinear invariants do not exist in many instances.
However, delineating exactly which instances admit separating semilinear invariants is challenging,
and is our main technical contribution on this front. The following few examples illustrate some of the phenomena that occur.

\begin{example}\label{ex:2}
    Remove the expanding factor from the previous instance, that is, put instead
    \[
    A = \frac 1 {\sqrt 5} \begin{bmatrix} 1 & -2 \\ 2 & 1\end{bmatrix}.
    \]
    Now $A$ being a rotation of an irrational angle, the orbit of $x$ is dense in the circle of radius 1.
    It is quite easy to prove that no semilinear invariant exists
    (except for the whole space $\R^2$) for this instance,
    whatever the value of $y$. This gives a first instance of non-existence of a semilinear invariants.
    Many such examples exist, and we shall now supply a more subtle one. Note that simple invariants do exist,
    such as the unit circle, which is a semialgebraic set but not a semilinear one.
\end{example}

\begin{example}\label{ex:3}
    Consider $\ell=(A,x,y)$ in dimension $4$ with
    \[
    A = \begin{bmatrix} A' & I_2 \\ 0 & A' \end{bmatrix},
    \]
    where $A'$ is the matrix from Example~\ref{ex:2}, $x=(0,0,1,0)$ and $y$ is arbitrary.
    When repeatedly applying $A$ to $x$, the last two coordinates describe a circle of radius 1 as in the previous example.
    However, the first two coordinates diverge: at each step, they are rotated and the last two coordinates are added.
    In this instance, no semilinear invariant exists (except again for
    the whole space $\R^4$), however proving this
    is somewhat involved. Note however once more that a semialgebraic invariant may easily be constructed.
\end{example}

In examples~\ref{ex:2} and~\ref{ex:3}, no non-trivial semilinear
invariant exist, or equivalently any semilinear invariant must contain
$\I_0$, where $\I_0$ is the whole space.
In all instances for which constructing an invariant is not
necessarily immediate (as is the case in Example~\ref{ex:1}),
we will provide a minimal invariant, that is, a semilinear $\I_0$ with
the property that any semilinear invariant will have to contain  $\I_0$.
In such cases there exists a semilinear invariant (namely $\I_0$) if and only if $y \notin \I_0$.
We conclude with two examples having such minimal semilinear invariants.

\begin{example}\label{ex:4}
    Consider $\ell=(A,x,y)$ in dimension $3$ with
    \[
    A = \begin{bmatrix} A' & 0 \\ 0 & -1 \end{bmatrix},
    \]
    where $A'$ is the matrix of Example~\ref{ex:2}, a $2$-dimensional rotation by an angle which is not a rational multiple of $2\pi$ and $x=(1,0,1)$.
    As we iterate matrix $A$, the two first coordinates describe a circle, and the third coordinate alternates between $1$ and $-1$:
    the orbit is dense in the union of two parallel circles.
    Yet the minimal semilinear invariant comprises the union of the two planes containing these circles.
\end{example}

\begin{figure}[ht]
    \centering
    \includegraphics[scale=0.5]{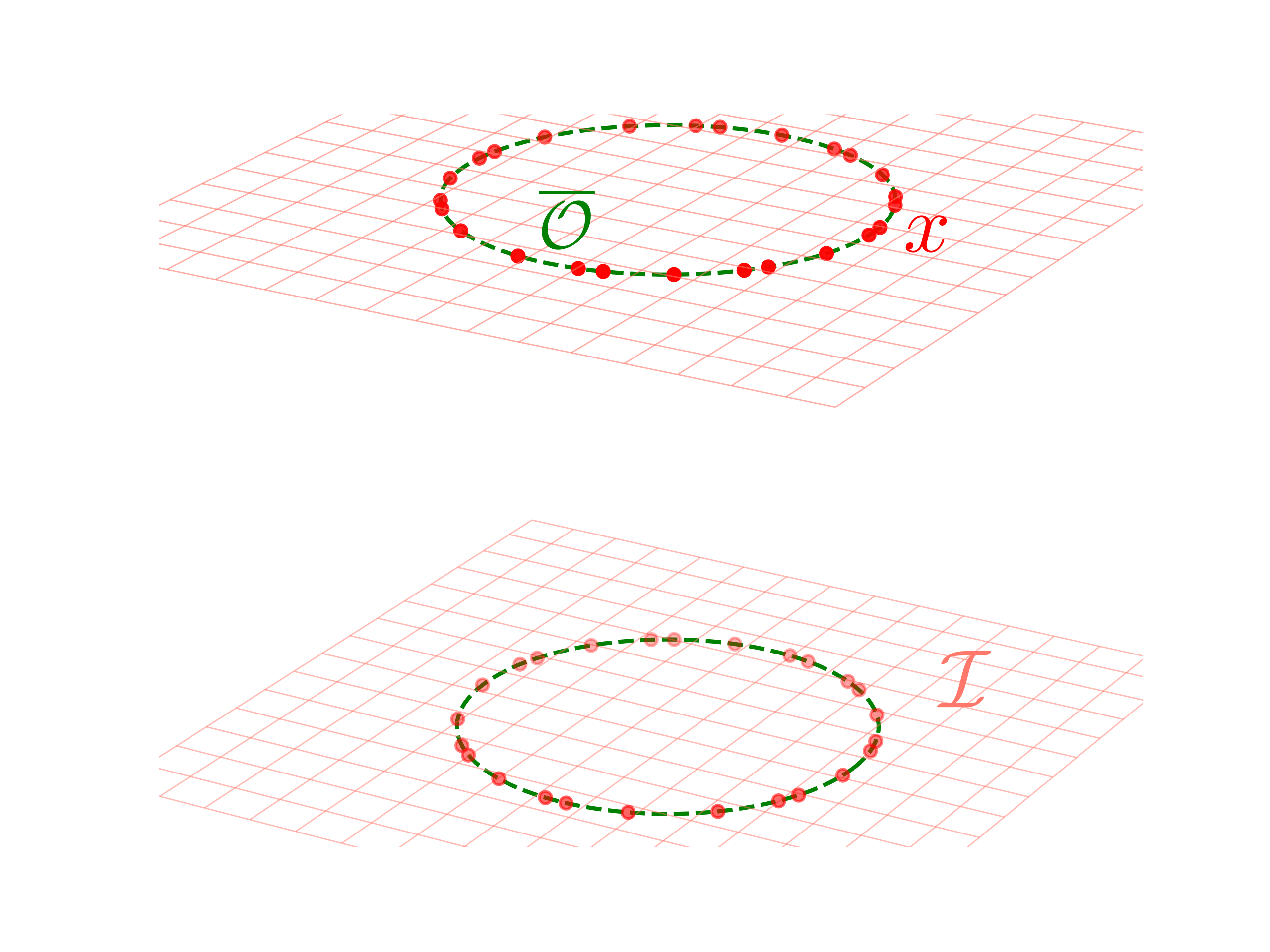}
    \caption{The minimal invariant for Example~\ref{ex:4}. Here, $\bar \O$ denotes the topological closure of the orbit of $x$.\label{fig:ex4}}
\end{figure}

\begin{example}\label{ex:5}
    Consider $\ell=(A,x,y)$ in dimension $8$ with
    \[
    A = \begin{bmatrix} A' & 0 \\ 0 & -A' \end{bmatrix},
    \]
    where $A'$ is the matrix from Example~\ref{ex:3}. This can be seen as two instances of Example~\ref{ex:3} running in parallel.
    Let $x=(0,0,1,0,0,0,-7,0)$, and note that both blocks of $x$ are initially related by a multiplicative factor, namely $-7(x_1,x_2,x_3,x_4) = (x_5, x_6, x_7, x_8)$. Moreover, as the first block is multiplied by the matrix $A'$ while the second one is multiplied by $-A'$, the multiplicative factor relating the two blocks alernates between $7$ and $-7$.
    Thus, the minimal semilinear invariant in this setting is
    \[
    \I_0 =  \{u \in \R^8 \mid (u_1,u_2,u_3,u_4) =  \rpm 7 (u_5,u_6,u_7,u_8)\},
    \]
    which has dimension 4. If however, we had $x=(0,0,1,0,1,0,-7,0)$, then the minimal semilinear invariant would be
    \[
    \{u \in \R^8 \mid (u_3,u_4) = \rpm 7 (u_7,u_8)\},
    \]
    which has dimension 6. Roughly speaking, no semilinear relation holds between $(u_1,u_2)$ and $(u_5,u_6)$.
\end{example}

\section{Undecidability proofs}\label{sec:undec}
\subsection{Proof of Theorem~\ref{th:closed-monniaux}}

\newcommand{\enc}[1]{\left[#1\right]}

We reduce an instance of the $\omega$-PCP problem defined as follows:
given nine pairs of non-empty words $\{(u^{(1)},v^{(1)}),\dots, (u^{(9)},v^{(9)})\}$ on alphabet $\{0,2\}$,
does there exist an infinite word $w=w_1w_2\dots$ on alphabet $\{1,\dots,9\}$ such that
$u^{(w_1)}u^{(w_2)}\dots  = v^{(w_1)}v^{(w_2)}\dots$ .
This problem is known to be undecidable~\cite{HH06}. 

In order to simplify future notations, given a finite or infinite word $w$,
we denote $|w|$ the length of the word $w$ and given an integer $i\leq |w|$,
we write $w_i$ for the $i$'th letter of $w$.
Given a finite or infinite word $w$ on alphabet $\{1,\dots,9\}$ we denote $u^{(w)}$ and 
$v^{(w)}$ the words on the alphabet $\{0,2\}$ such that $u^{(w)}= u^{(w_1)}u^{(w_2)}\dots$ and 
$v^{(w)}= v^{(w_1)}v^{(w_2)}\dots$. Given a (finite or infinite) word $w$ on 
the alphabet $\{0,2\}^*$, denote by
$\enc{w}=\sum_{i=1}^{|w|}w_i4^{1-i}$ the quaternary encoding of $w$. It is
clear that it satisfies $\enc{ww'}=\enc{w}+4^{-|w|}\enc{w'}$ and that $\enc{w}\in[0,\tfrac{8}{3}]$.

\bigskip
Let $\{(u^{(1)},v^{(1)}),\dots, (u^{(9)},v^{(9)})\}$ be an instance of the $\omega$-PCP problem.
For all $i\leq 9$, for readibility, we denote $|u^{(i)}|={n_i}$ and $|v^{(i)}|={m_i}$.
We build the matrices $M_1,\dots, M_9$ where
\[
    M_i=\begin{bmatrix}
        1   &   \enc{u^{(i)}}   &   -\enc{v^{(i)}}  \\
        0   &   4^{-n_i}    &   0           \\
        0   &   0           &   4^{-m_i}
    \end{bmatrix}
\]
In the following, we write $M_w$ for $w=w_1\dots w_k\in \{1,\dots,9\}^*$ the matrix $M=M_{w_k}\dots M_{w_1}$,
which can be checked to satisfy
\[
    M_w=\begin{bmatrix}
        1   &   \enc{u^{(w)}}   &   -\enc{v^{(w)}}  \\
        0   &   4^{-|u^{(w)}|}  &   0           \\
        0   &   0           &   4^{-|v^{(w)}|}
    \end{bmatrix},
    \qquad
    M_w\begin{bmatrix}0\\1\\1\end{bmatrix}=\begin{bmatrix}\enc{u^w}-\enc{v^w}\\4^{-|u^w|}\\4^{-|v^w|}\end{bmatrix}.
\]
Let us show that there exists a separating invariant of $((M_i)_{i\leq 9},x,y)$
where $x=(0,1,1)$ and $y=(0,0,0)$
iff the $\omega$-PCP instance has no solution. 

\bigskip

Let us first assume the $\omega$-PCP instance has a solution $w$. 
Fix $r\in \N$ and let $w\restriction_r=w_1\cdots w_r$ and $x_r= M_{w\restriction_r} x$.
We have that $x_r= (\enc{u^{(w\restriction_r)}}-\enc{v^{(w\restriction_r)}},$ $4^{-|u^{(w\restriction_r)}|},4^{-|v^{(w\restriction_r)}|})$ and
since $u^{(w)}=v^{(w)}$, it is clear that $x_r\to0=y$ as $r\to\infty$.
Any separating invariant $\I$ must contain
this sequence $x_r$ since $\I$ contains the initial point and is stable under
$(M_i)_{i\leq 9}$. Moreover, $\I$ is closed so it must contain the limit of the sequence,
$(0,0,0)$, which is the target point. Thus $\I$ cannot be a separating invariant.
Therefore there is no separating invariant of $((M_i)_{i\leq 9},(0,1,1),(0,0,0))$.

\bigskip
Now, let us assume the $\omega$-PCP instance has no solution. 
There exists $n_0\in \nats$ such that for every infinite word $w$ on alphabet
$\{0,\dots, 9\}$ there exists $n\leq n_0$ such that $u^{(w)}_n \neq v^{(w)}_n$.
Indeed, consider the tree which root is labelled by $(\eps,\eps)$ and, given a node $(u,v)$ of the tree,
if for all $n\leq \min(|u|,|v|)$ we have $u_n=v_n$, then this node has $9$ children:
the nodes $(u u^{(i)},v v^{(i)})$ for $i=1\dots 9$. This tree is finitely branching and
does not contain any infinite path (which would induce a solution to the $\omega$-PCP instance).
Thus, according to K\"onig's lemma, it is finite. We can therefore choose the height of this tree as our $n_0$.

We define the invariant $\I=\I'\cup \I''$ where\footnote{This is a semilinear invariant since $|x|\geqslant y$
    if and only if $x\geqslant y\vee -x\geqslant y$.}
\[\I'= \big\{(s,c,d): |s| \geq 4(c+d)+4^{-n_0-1}\wedge c\geq 0 \wedge d \geq 0 \big\}\]
and
\[\I''=\big\{M_w x : w\in \{1,\dots,9\}^* \wedge |w|\leq n_0+1\big\}\]
Let us show that $\I$ is a separating invariant of $((M_i)_{i\leq 9},(0,1,1),(0,0,0))$.
By definition, $\I$ is closed, semilinear, contains $x$ and does not contain $y$.
The difficult point is to show stability under $M_i$ for $i\leq 9$.
\begin{itemize}
\item Let $M_wx\in \I''$, for some $w$: there are two cases. Either $|w|\leqslant n_0$,
    then $|wi|\leqslant n_0+1$, therefore $M_iM_wx=M_{wi}z\in \I''$. Otherwise,
    $M_iM_wx=M_{wi}x=(s,c,d)$ where $s=\enc{u^{(wi)}}-\enc{v^{(wi)}}$, $c=4^{-|u^{(wi)}|}$ and $d=4^{-|v^{(wi)}|}$.
    But then, there exists $n\leqslant n_0$ such that $u^{(wi)}_n\neq v^{(wi)}_n$. Let $n$ be
    the smallest such number, then
    \begin{align*}
        s
            &= \enc{u^{(wi)}}-\enc{v^{(wi)}}\\
            &= (u^{(wi)}_n-v^{(wi)}_n)4^{1-n}+\sum_{j=n+1}^{|wi|}(u^{(wi)}_j-v^{(wi)}_j)4^{1-j}    
    \end{align*}            
\text{since }$u^{(wi)}_j=v^{(wi)}_j$\text{for }$j<n$. Thus,
\begin{align*}
        |s| &\geqslant 2\cdot 4^{1-n}-\tfrac{8}{3}4^{-n}   &&\text{since }|u^{(wi)}_n-u^{(wi)}_n|=2\text{ and }\enc{\cdot}\in[0,\tfrac{8}{3}]\\
            &\geqslant 4^{1-n}+4^{-n}\\
            &\geqslant 4(c+d)+4^{-n_0-1}&&\text{since }n\leqslant n_0\text{ and }|u^{(wi)}|,|v^{(wi)}|\geqslant n_0+2.
    \end{align*}
    This shows that $M_i(M_wx)\in \I'\subseteq \I$.
\item Let $z=(s,c,d)\in \I'$, then $|s| \geqslant  4(c+d)+4^{-n_0-1}$. Without loss of generality,
    assume that $d\geqslant c$ (this is completely symmetric in $c$ and $d$). Let
    $(s',c',d')=M_iz$, and we check that
    then
    \begin{align*}
        |s'|
            &= |s+c\enc{u^{(i)}}-d\enc{v^{(i)}}|&&\text{by applying the matrix }M_i\\
            &\geqslant |s|-d\max(\enc{u^{(i)}},\enc{v^{(i)}})\\
            &\geqslant 4(c+d)+4^{-n_0-1}-d\max(\enc{u^{(i)}},\enc{v^{(i)}})&&\text{by assumption on }s\\
            &\geqslant 4(c+d)+4^{-n_0-1}-d\tfrac{8}{3}&&\text{since }\enc{\cdot}\in[0,\tfrac{8}{3}]\\
            &= 4(c+ d/3) + 4^{-n_0-1}  \\
            &\geqslant 4(c'+d') + 4^{-n_0-1}&&\text{since }c\geqslant c'\text{ and }d/4\geqslant d'
    \end{align*}
    since $c'=c4^{-|u^{(i)}|}$ and $d'=d4^{-|v^{(i)}|}$.
    This shows that $M_iz\in \I'\subseteq \I$.
\end{itemize}
This shows that $\I$ is thus stable and concludes the reduction.

\subsection{Proof of Theorem~\ref{th:closed-monniaux-bis}}

We reduce the instances of Theorem~\ref{th:closed-monniaux} to $2$ matrices of size $27$. The first matrix
$M_s$ shifts upwards the position of the values in the point by 3, while the second matrix $M_p$ applies one of the matrices
of the previous reduction, depending on the position of the values within the matrices, then put the obtained value at the top. 
In other words, $M_pM_s^{i-1}$ for $1\leq i\leq 9$ intuitively has the same effect as $M_i$ had in the proof of Theorem~\ref{th:closed-monniaux}.
In the following, we reuse the notations and results of the proof of Theorem~\ref{th:closed-monniaux}.

Define matrices $M_s$ and $M_p$, where $I_3$ is the identity matrix of size $3\times 3$,
and for any $z\in\R^3$ and $i\in\set{0,\ldots,8}$, the $i^{th}$ shift $z^{\downarrow i}\in\R^{27}$ of $z$,
where $\boldsymbol{0}_{n}\in\R^n$ denotes the zero vector of size $n$, as follows:
\[
    M_s=\begin{bmatrix}
        0 & \cdots & 0 & I_3 \\
        I_3 & & \\[-4pt]
        & \ddots & & \\[-4pt]
         &  & I_3 & 0 \\
        \end{bmatrix},
    \qquad
    M_p=\begin{bmatrix}
        M_1 & \cdots  & M_9\\
        0 & \cdots & 0\\[-4pt]
        \vdots & & \vdots \\
        0 & \cdots & 0\\[1pt]
        \end{bmatrix},
    \qquad
    z^{\downarrow i}=\begin{bmatrix}\boldsymbol{0}_{3i}\\z\\\boldsymbol{0}_{24-3i}\\[1pt]\end{bmatrix}.
\]
It follows that $M_sz^{\downarrow i}=z^{\downarrow i+1\bmod 9}$ and $M_pz^{\downarrow i}=(M_{i+1}z)^{\downarrow 0}$.
Assume that there exists a separating invariant $\I$ for $(M_1,\ldots,M_9,x,y)$
and let
\[\J=\bigcup_{i=0}^8\set{z^{\downarrow i}:z\in \I}\]
which is a closed semilinear set. Then for any $z^{\downarrow i}\in \J$, we have
$M_sz^{\downarrow i}=z^{\downarrow i+1\bmod 9}\in \I$ by definition and
$M_pz^{\downarrow i}=(M_iz)^{\downarrow 0}\in \J$ since $M_iz\in \I$ by
virtue of $z\in \I$ and $\I$ being invariant. Furthermore, $x'=x^{\downarrow 0}\in \I$
since $x\in \I$, and $y'=y^{\downarrow 0}\notin \J$ for otherwise we would have $y\in \I$.
Therefore $\J$ is a separating invariant for  $(M_s,M_p,x',y')$.

Assume that there exists a separating invariant $\J$ for $(M_s,M_p,x',y')$
and let $\I=\set{z:z^{\downarrow 0}\in \J}$ which is a closed semilinear set. Clearly $x\in \I$ since $x'=x^{\downarrow 0}\in \J$
and $y\notin \I$ since $y'=y^{\downarrow 0}\notin \J$. Let $z\in \I$ and $i\in\set{1,\ldots,9}$,
then $(M_iz)^{\downarrow 0}=M_pM_s^{i-1}z^{\downarrow 0}\in \J$ and since $z^{\downarrow 0}\in \J$
and $\J$ is invariant under $M_s$ and $M_p$, thus $M_iz\in \I$. Therefore $\I$ is a non-reachability
invariant for $(M_1,\ldots,M_9,x,y)$.

\subsection{Proof sketch of Theorem~\ref{th:closed-monniaux-robust}}

We do the proof of Theorem~\ref{th:closed-monniaux-robust} twice:
first we use linear guards in order to limit the selection of the matrices.
The added power of the guards allows for a relatively simple proof. 
This first proof can be seen as an extended sketch of the
second one, in Appendix~\ref{appendix:undec}, where we remove the guards to obtain the result claimed. We do so
by emulating the guards using extra variables.
\bigskip

We reduce from the $\omega$-PCP problem and reuse some of the notations of the proof of Theorem~\ref{th:closed-monniaux}.
Let $\{(u^{(1)},v^{(1)}),\dots, (u^{(9)},v^{(9)})\}$ be an instance of the $\omega$-PCP problem.
We build the matrices $\hat{M}_1,\dots, \hat{M}_9,M_e,M_-$ where
\[
    \hat{M}_i=\begin{bmatrix}
        M_i & & \\
        & 1 & 2 \\
        & 0 & 1
    \end{bmatrix},
    \qquad
    M_e=\begin{bmatrix}
        \boldsymbol{0}_{3\times3} & & \\
        & 1  & 0 \\
        & 0 & 1
        \end{bmatrix},
    \qquad
    M_-=\begin{bmatrix}
        I_3 & & \\
        & 1 & -2 \\
        & 0 & 1
        \end{bmatrix}
\]
and $M_1,\ldots,M_9$ are from the proof of Theorem~\ref{th:closed-monniaux}.
Moreover, when in $(s,c,d,n,a)$, the matrices $\hat{M}_i$ and $M_e$ can only be selected
if the linear guard $|s| < 4(c+d)$ holds, and the matrix $M_-$ can only be selected if $s=c=d=0$.
\smallskip

Informally, in state $(s,c,d,n,a)$,
the subvector $(s,c,d)$ has the same role as before: $s$ contains the difference of the
values of the numbers obtained using the $v_i$ and $u_i$, while $c$ and $d$ are used in
order to help compute this value.
In the proof of Theorem~\ref{th:closed-monniaux}, we showed that when the $\omega$-PCP
instance had no solution, there existed a value $n_0$ such that any pair of words created
with the alphabet $(u^{(i)},v^{(i)})$ differed on one of the first $n_0$ terms. 
The variable $n$ is used with the guards in order to detect this value $n_0$:
if such an $n_0$ exists, then at most $n_0+1$ matrices $M_i$ can be selected before the guard stops holding.
Moreover, firing a matrix $M_i$ adds 2 to $n$ ensuring that when the guard stops holding,
$n$ is smaller or equal to $2(n_0+1)$.
Conversely, if no such $n_0$ exist, then there is a way to select matrices $M_i$ such that the guard always holds,
allowing the variable $n$ to become an even number as high as one wants.
The existence of an upper bound on the value of $n$ is used to build an invariant or to prove that there cannot exist an invariant.
Finally, the value $a$ is only here in order to allow for affine modification of the values. It is never modified.
\smallskip

Let $\hat{x}=(x,0,1)$ and $\hat{y}=(y,1,1)$.
Note that $\hat{y}$ is not in the adherence of the reachable set as the fourth variable of any
reachable point is an even number while $y$'s is an odd one. 

Assume the $\omega$-PCP instance does not possess a solution.
Then there exists $n_0\in \nats$ such that any pair of words
$(u^{(w)},v^{(w)})$ differs on one of the first $n_0$ letters. 
Define the invariant $\I=\I'\cup \I''$ where
\begin{align*}
    \I' &= \{\hat{M}_w\hat{x} : w\in \{1,\dots,9\}^* \wedge |w|\leqslant n_0+1\}\\
    \I''&=\{(0,0,0,n,1): n\leqslant 0 \vee (\exists k \in \nats, n=2k \wedge n \leqslant 2(n_0+1))\}.
\end{align*}
This invariant is clearly semilinear, it contains $\hat{x}$ and does not contain $\hat{y}$.
If $z=(0,0,0,n,1) \in \I''$ then only $M_-$ can be triggered due to the guards
and $M_-z=(0,0,0,n-2,1)\in \I''$. Now if $z=(s,c,d,n,a)=\hat{M}_w\hat{x} \in \I'$
for some $w\in \{1,\dots,9\}^*$, then
$M_-$ cannot be fired as the guard does not hold. If one fires $M_e$, by construction of $\I'$, 
$n$ is an even number smaller than $2(n_0+1)$, thus 
$M_ez\in \I''$. Now in order to fire a matrix $\hat{M}_i$, one needs 
$|s| < 4(c+d)$ to hold.
We showed in the proof of Theorem~\ref{th:closed-monniaux} that, from the initial configuration $x$,
after $n_0+1$ transitions using one of the matrices $M_i$ then $1/4^{n_0+1}\leqslant |s| - 4(c+d)$. 
As a consequence, if the guard holds, then $|w|\leqslant n_0$ and  
$\hat{M}_iz= \hat{M}_{wi}\hat{x}\in \I'$.
Therefore, $\I$ is a separating invariant of $(\hat{M}_1,\dots \hat{M}_9, M_e, M_-,\hat{x},\hat{y})$.

Now assume the $\omega$-PCP possesses a solution $w\in \{1,\dots,9\}^\omega$.
For $k\in\nats$, we denote $w\restriction_k$ the prefix of length $k$ of $w$.
Let $k\in \nats$ and $(s,c,d,n,a) = \hat{M}_{w\restriction_k} x$,
then $|s|<4(c+d)$. 
Indeed, assume that $u^{(w\restriction_k)}$ is longer than
$v^{(w\restriction_k)}$. Then $u^{(w\restriction_k)}=v^{(w\restriction_k)}t$ for some word $t\in\set{0,2}^*$
because $u^w{(w)}v^{(w)}$. Let $\ell=|u^{(w\restriction_k)}|$ and recall that $c=4^{-\ell}$, then
\[
    s=|[u^{(w\restriction_k)}]-[v^{(w\restriction_k)}]|
        =4^{-\ell}[t]
        \leqslant 4^{-\ell}\tfrac{8}{3}
        \leqslant 4c
        < 4(c+d).
\]
The symmetric case is similar but uses $d$ instead. Therefore the guard is satisfied and
$M_e\hat{M}_{w\restriction_k}\hat{x}=(0,0,0,2k,1)\in \I$
is reachable for all $k\in\N$.
Let $\I$ be a semilinear invariant containing the reachability set,
then $\I \cap \{(0,0,0,x,1): x\in \mathbb{R}\}$ is semilinear and contains
$(0,0,0,2k,1)$ for all $k\in \nats$. This implies that it necessarily contains
an unbounded interval and there must exists $k_0\in\nats$ such that  
$(0,0,0,2k_0+1,1) \in \I$. Since $\I$ is stable by the matrix $M_-$, $\I$ contains 
the target $y$.
Therefore, $\I$ is not a separating invariant of $((\hat{M}_1,\dots \hat{M}_9, M_e, M_-),x,y)$.

\section{Decidability proofs}\label{sec:dec}

This section is aimed at sketching the main ideas of the proof of Theorem~\ref{thm:main} while avoiding technicalities and details. We point to the appendix for full proofs.
Recall that we only consider closed semilinear invariants.

\begin{itemize}
\item We first normalize the Orbit instance, which amounts to putting matrix $A$ in Jordan normal form,
and eliminating some easy instances. This is described in Section~\ref{sec:normalization}.
\item We then eliminate some positive cases in Section~\ref{sec:positive}.
More precisely, we construct invariants whenever one of the three following conditions is realized:
\begin{itemize}
\item $A$ has an eigenvalue of modulus $>1$.
\item $A$ has an eigenvalue of modulus $<1$.
\item $A$ has a Jordan block of size $\geq 2$ with an eigenvalue that is a root of unity.
\end{itemize}
\item We are now left with an instance where all eigenvalues are of modulus 1 and not roots of unity,
    which is the most involved part of the paper. In this setting, we exhibit the minimal semilinear invariant $\I$ containing $x$.
    In particular, there exists a semilinear invariant (namely, $\I$) if and only if $y \notin \I$. This part is explained in Section~\ref{sec:minimal}.
\end{itemize}

\subsection{Normalization}\label{sec:normalization}
As a first step, recall that every matrix $A$ can be written in the
form $A = Q^{-1} J Q$, where $Q$ is invertible and $J$ is in Jordan
normal form.  The following lemma transfers semilinear invariants
through the change-of-basis matrix $Q$.

\begin{lemma}\label{lem:basis}
Let $\ell = (A,x,y)$ be an Orbit instance, and $Q$ an invertible matrix in $\A^{d \times d}$.
Construct the Orbit instance $\ell_Q = (Q A Q^{-1},Q x, Q y)$.
Then $\I$ is a semilinear invariant for $\ell_Q$
if, and only if, $Q^{-1} \I$ is a semilinear invariant for $\ell$.
\end{lemma}

\begin{proof}
First of all, $Q^{-1} \I$ is semilinear if, and only if, $\I$ is semilinear.
We have:
\begin{itemize}
    \item $Q A Q^{-1} \I \subseteq \I$ if, and only if, $A Q^{-1} \I \subseteq Q^{-1} \I$,
  	\item $Q x \in \I$ if, and only if, $x \in Q^{-1} \I$,
	\item $Q y \notin \I$, if, and only if, $y \notin Q^{-1} \I$.
\end{itemize} 
This concludes the proof.
\end{proof}

Thanks to Lemma~\ref{lem:basis}, we can reduce the problem of the
existence of semilinear invariants for Orbit instances to cases
in which the matrix is in Jordan normal form, \emph{i.e.}, is a diagonal
block matrix, where the blocks (called Jordan blocks) are of the form:
\[
\begin{bmatrix}
\lambda & 1            & \;     & \;  \\
\;        & \lambda    & \ddots & \;  \\
\;        & \;         & \ddots & 1   \\
\;        & \;         & \;     & \lambda       
\end{bmatrix}
\]
Note that this transformation can be achieved in polynomial time~\cite{Cai00,CLZ00}. 
Formally, a Jordan block is a matrix $\lambda I + N$ with $\lambda \in
\C$, $I$ the identity matrix and $N$ the matrix with $1$'s on the
upper diagonal, and $0$'s everywhere else.  The number $\lambda$ is an
eigenvalue of $A$. We will use notation $\Jord{d}{\lambda}$ for the Jordan block of size $d$ with eigenvalue $\lambda$. A Jordan block of dimension one is called
diagonal, and $A$ is diagonalisable if, and only if, all Jordan blocks
are diagonal.

\vskip1em
The $d$ dimensions of the matrix $A$ are indexed by pairs $(J,k)$,
where $J$ ranges over the Jordan blocks and $k \in
\set{1,\ldots,d(J)}$ where $d(J)$ is the dimension of the Jordan
block $J$.  For instance, if the matrix $A$ has two Jordan blocks,
$J_1$ of dimension $1$ and $J_2$ of dimension $2$, then the three
dimensions of $A$ are $(J_1,1)$ (corresponding to the Jordan block
$J_1$) and $(J_2,1),(J_2,2)$ (corresponding to the Jordan block
$J_2$).

For a point $v$ and a subset $S$ of $\set{1,\ldots,d}$, let $v_S$ be the projection of $v$ on the dimensions in $S$,
and extend this notation to matrices.
For instance, $v_J$ is the point corresponding to the dimensions of the Jordan block $J$,
and $v_{J,>k}$ is projected on the coordinates of the Jordan block $J$ whose index is greater than $k$.
We write $\compl S$ for the coordinates which are not in $S$.

\medskip
There are a few degenerate cases which we handle now.
We say that an Orbit instance $\ell = (A,x,y)$ in Jordan normal form is normalized if:
\begin{itemize}
	\item There is no Jordan block associated with the eigenvalue $0$, or equivalently $A$ is invertible.
	\item For each Jordan block $J$, the last coordinate of the point $x_J$ is not zero,
	\textit{i.e.} $x_{J,d(J)} \neq 0$.
	\item There is no diagonal Jordan block with an eigenvalue which is a root of unity,
	\item Any Jordan block $J$ with an eigenvalue of modulus $<1$ has $y_J \neq 0.$
\end{itemize}

\begin{lemma}\label{lem:reduc}
The existence of semilinear invariants for Orbit instances reduces to the same problem for normalized Orbit instances in Jordan normal form.
\end{lemma}

Lemma~\ref{lem:reduc} is proved in Appendix~\ref{appendix:reduc}.

\subsection{Positive cases}\label{sec:positive}
Many Orbit instances present a divergence which we can exploit to construct a semilinear invariant. Such behaviours are easily identified once the matrix is in Jordan Normal Form, as properties of its Jordan blocks. We isolate three such cases.
\begin{itemize}
\item If there is an eigenvalue of modulus $>1$. Call $J$ its Jordan block. Projecting to the last coordinate of $J$ the orbit of $x$ diverges to $\infty$ in modulus (see Example~\ref{ex:1}). A long enough ``initial segment'' $\{x, Ax, \dots, A^k x\}$ together with the complement of its convex hull (on the last coordinate of $J$) constitutes a semilinear invariant. See Appendix~\ref{appendix:more_one} for details.

\item If there is an eigenvalue of modulus $<1$ in block $J$, the situation is quite similar with a convergence towards 0. However, the construction we give is more involved, the reason being that we may not just concentrate on the last nonzero coordinate $x_{J,l}$ of $x_J$, since $y_{J,l}$ may very well be 0, which belongs to the adherence of the orbit on this coordinate. Yet on the full block, $y_J \neq 0$. We show how to construct, for any $0 < \varepsilon$, a semilinear invariant $\I$ such that $B(0,\varepsilon') \subseteq \I \subseteq B(0, \varepsilon)$ for some $\varepsilon'>0$. Picking $\varepsilon$ small enough we make sure that $y \notin \I$, and then $\{x,Ax, \dots, A^k x\} \cup \I$ is a semilinear invariant if $k$ is large enough so that $|| A^{k}x || \leq \varepsilon'$. See Appendix~\ref{appendix:less_one} for more details.

\item Finally, if there is an eigenvalue which is a root of unity, say $\lambda^n=1$, on a Jordan block $J$ of size at least 2 (that is, a non diagonal block), then penultimate coordinate on $J$ of the orbit goes to $\infty$ in modulus. In this case, the orbit on this coordinate is contained in a union of $n$ half-lines which we cut far enough away from 0 and add an initial segment to build a semilinear invariant. See Appendix~\ref{appendix:mod_one_rou_not_diag} for details.
\end{itemize}

Note that in each of these cases, we concentrate on the corresponding (stable) eigenspace, construct a separating semilinear invariant for this restriction of the problem, and extend it to the full space by allowing any value on other coordinates.

\subsection{Minimal invariants}\label{sec:minimal}
We have now reduced to an instance where all eigenvalues have modulus 1 and are not roots of unity. Intuitively, in this setting, semilinear invariants fail, as they are not precise enough to exploit subtle multiplicative relations that may hold among eigenvalues. 
However, it may be the case that some coarse information in the input can still be stabilised by an semilinear invariant, for instance if two synchronised blocks are exactly identical (see Examples~\ref{ex:4} and \ref{ex:5} for more elaborate cases).

We start by identifying exactly where semilinear invariants fail. Call two eigenvalues equivalent if their quotient is a root of unity (that is, they have a multiplicative relationship of degree 1). We show that whenever no two different eigenvalues are even non-equivalent, the only stable semilinear sets are trivial. As a consequence, computing the minimal semilinear invariant in this setting is easy, as it is basically the whole space (except where $x$ is 0). However, this lower bound (non-existence of semilinear invariant) constitutes the most technically involved part. Our proof is inductive with as base case the diagonal case, where it makes crucial use of the Skolem-Mahler-Lech theorem. This is the subject of Appendix~\ref{subsec:mod_one_not_rou_not_eq}.

When the matrix has several equivalent eigenvalues, we show how to iteratively reduce the dimension in order to eventually fall into the previous scenario. Rougly speaking, if $A$ is comprised of two identical blocks $B$, we show that it suffices to compute a minimal invariant $\I_B$ for $B$, since $\{z \mid \tilde z_1 \in \I_B \text{ and } \tilde z_2 = \tilde z_1\}$ (with obvious notations) is a minimal invariant for $A$. This is achieved, by first assuming that all equivalent eigenvalues are in fact equal and then easily reducing to this case by considering a large enough iterations of $A$, in Appendix~\ref{subsec:mod_one_not_rou_not_eq}.

\bibliography{biblio}

\appendix

\section{Proof of Theorem~\ref{th:closed-monniaux-robust}}\label{appendix:undec}
We now turn to the second proof without the use of guards.
The idea is similar, however, the reduction is far more involved as the use of a guard is now replaced by the test of reachability of additional variables.

We reduce an instance of the $\omega$-PCP problem.
Let $\{(u^{(1)},v^{(1)}),\dots, (u^{(9)},v^{(9)})\}$ be an instance of the $\omega$-PCP problem.  We denote $|u^{(i)}|={n_i}$ and $|v^{(i)}|={m_i}$ and
we build the matrices $M_1,\dots M_9, M^1_m, M^2_m,$ $M_x, M_e, M_-, M_p,M_d$ of dimension 21 where:


\begin{itemize}
        \item for $i=1\dots 9$, $M_i[1,1]=1, M_i[1,2]=\enc{u^{(i)}}=(u^{(i)}_1 + u^{(i)}_2/4 + \dots + u^{(i)}_{n_i}/4^{{n_i}-1})$, $M_i[1,3]=-\enc{v^{(i)}}$, $M_i[2,2]=1/4^{n_i}$, $M_i[3,3]=1/4^{m_i}$, $M_i[4,4]=1$, $M_i[4,7]=2$, $M_i[5,5]=1$,  $M_i[6,6]=4^{n_i + m_i}$ $M_i[7,7]=1$ and all other values are set to 0.

        \item $M^1_m[8,1]=1$, $M^1_m[8,2]=-4=M^1_m[8,3]$,
        $M^1_m[9,1]=1$, $M^1_m[9,2]=-4=M^1_m[9,3]$, $M^1_m[9,5]=-1$,
        $M^1_m[10,1]=1$, $M^1_m[11,4]=1$, $M^1_m[12,1]=1$,  $M^1_m[13,6]=1$, $M^1_m[14,7]=1$ and all other values are set to 0.

        \item $M^2_m[8,1]=-1$, $M^2_m[8,2]=-4=M^2_m[8,3]$, $M^2_m[8,5]=-1$,
        $M^2_m[9,1]=-1$, $M^2_m[9,2]=-4=M^2_m[9,3]$, $M^2_m[9,5]=-1$,
        $M^2_m[10,1]=-1$, $M^2_m[11,4]=1$, $M^2_m[12,1]=-1$,  $M^2_m[13,6]=1$, $M^2_m[14,7]=1$ and all other values are set to 0.

        \item  for $i=8\dots 14$, $M_x[i,i]=1$, $M_x[8,9]=M_x[10,12]=1$, $M_x[13,14]=-1$ and all other values are set to 0.

        \item $M_e[15,8]=M_e[17,10]=M_e[18,11]=M_e[20,13]=M_e[21,14]=1$ and all other values are set to 0.

        \item $M_-[18,21]=-2$, for $i=15\dots 21$, $M_-[i,i]=1$ and all other values are set to 0.

        \item $M_p[15,21]=1$, for $i=15\dots 21$, $M_p[i,i]=1$ and all other values are set to 0.

        \item $M_d[17,21]=-1$, for $i=15\dots 21$, $M_d[i,i]=1$ and all other values are set to 0.

\end{itemize}

Let us explain informally the effects of the different matrices. 
The dimensions can be separated in three blocs of 7 each. The seven variables associated to the bloc number $i\in \{1,2,3\}$ are called $(s^i,c^i,d^i,n^i,r^i,k^i,a^i)$. 
As a general overview, the first bloc is used to apply the matrices $M_i$ as in the previous proof (with the guards), the second bloc ensures every value is an integer allowing us, in the third bloc, to use reachability tests in order to verify that the guard of the previous proof holds.

\begin{itemize}
        \item The matrices $M_i$ only affects the first bloc of variables. On these variables, the effect of these matrices is similar to the effect of the matrices $\tilde{M_i}$ of the extended sketch of proof of the Theorem~\ref{th:closed-monniaux-robust}. The one difference comes from the additional variables $r^1$, which is not used but kept for symmetry with the other blocs, and $k^1$ which is increased strongly by $M_i$ in a way ensuring that $k^1 s^1, k^1c^1$ and $k^1 d^1$ are integers.

        \item The matrices $M^1_m$ and $M^2_m$ modifies the variables of the bloc 1 and transfer them to the second bloc. Precisely, we have that applying $M_m^1$ realise the following transformation
        \[
        (s^2,c^2,d^2,n^2,r^2,k^2,a^2) = (s^1-4(c^1+d^1),s^1-4(c^1+d^1) ,s^1,n^1,s^1,k^1,a^1)
        \]
        and the variables of the bloc 1 are set to 0. $M_m^2$ has a similar effect, replacing instances of $s^1$ by $-s^1$.

        \item The matrix $M_x$ affects the second bloc this way:
        \[
           M_x(s^2,c^2,d^2,n^2,r^2,k^2,a^2) = (s^2+c^2,c^2,d^2+r^2,n^2,r^2,k^2-1,a^2).
        \]
        Remark that, from the values 
        $(s^2,s^2,d^2,n^2,d^2,k^2,a^2)$ reached after firing one of the matrix $M_m^i$, 
        using $M_x$ successively $k^2$ times reaches the values 
        $(k^2s^2,s^2,k^2d^2,n^2,d^2,0,a^2)$. This ensures the first and third variables are now integers.

        \item $M_e$ moves the variable to the third bloc while putting to 0 two variables that are not useful any more $c^3$ and $r^3$ (remark thus than one could simply use matrices of dimension 19 instead of 21. We kept 21 to simplify the representation). In other words, after applying $M_e$ we have
        \[
        (s^3,c^3,d^3,n^3,r^3,k^3,a^3) = (s^2,0,d^2,n^2,0,k^2,a^2)
        \]
        and all other variables are set to 0.

        \item The matrices $M_p$, $M_d$ and $M_-$ respectively increases by 1, decreases by 1 and decreases by 2 the value of $s^3$, $d^3$ and $n^3$.
\end{itemize}

The use of these three blocs divides the firing of a sequence of matrices in three steps. 
The first step corresponds to the firing of the matrices $M_i$ which are the only ones keeping the values within the first dimension bloc. This first step is similar to what was done in the previous proof. 
Using one of the two matrices $M^i_m$, $i\in\{1,2\}$, starts the second step by moving the values to the second bloc (the use of two different matrices is required to represent and check the absolute value in the guard of the previous proof). From now on, the matrices $M_i$ have no effect on the system.
The goal of the second, and partially the third, step is to verify whether the guard of the previous proof held when the second step started. 
To do so, we first use the matrix $M_x$ in order to turn the first two variables into integers. If $M_x$ is not applied an adequate number of times, then the variable labelled by $k$ will never be equal to $0$, ensuring non-reachability of the target. 
We then use the matrix $M_e$ to start the third step. 
Assuming the multiplication was correctly done, then $s^3$ and $d^3$ are now integers. Moreover, 
$s^3<0$ and $d^3 >0$ iff the guard of the previous proof held when the first step ended. Thus reaching the target on all variables, except for $n^3$, is possible iff the guard was indeed respected. The case of the counter $n^3$ is then identical to the one of $n$ in the previous proof.

Let us formally show that there exists a separating invariant of the instance
$((M_1,\dots M_9, M^1_m, M^2_m, M_x, M_e, M_-, M_p,M_d),x,y)$ with $x\in\{(0,1,1,0,1,1,1)\}\times\{0\}^{14} $ and $y\in\{0\}^{14}\times \{(-1,0,0,1,0,0,1)\}$ iff the $\omega$-PCP has no solution.

Assume first that the $\omega$-PCP does not possess a solution. There thus exists $n_0\in \nats$ such that any pair of words $(u^{(w)},v^{(w)})$ on $\{0,2\}$ generated from the $\omega$-PCP instance differs on one of the first $n_0$ letters.

We define the invariant $I=I_1\cup I_2 \cup I_3$ where

\begin{itemize}
        \item $I_1 = \{M_w x \mid w\in \{1,\dots,9\}^* \wedge |w|\leq n_0+1\}\cup  \{(s^1,c^1,d^1,n^1,r^1,k^1,1)\mid |s^1| \geq 4(c^1+d^1)\}\times\{0\}^{14}$.

        \item 
        $I_2 = \{M^j_xM_m^iM_w   x \mid i\in \{1,2\}\wedge j\in \N\wedge w\in \{1,\dots,9\}^* \wedge |w|\leq n_0+1\}\cup  
        \{0\}^{7}\times\{(s^2,c^2,d^2,n^2,r^2,k^2,1)\mid  k^2\leq -1 \vee (d^2< 0\wedge r^2< 0) \vee (s^2 \geq 0 \wedge c^2\geq 0)\}\times\{0\}^{7}$. \\
   $I_2$ is semilinear: while the first set is technically an infinite set of points, when $j$ is high enough, the variable $k^2$ becomes negative, ensuring that the point belongs to the second set of $I_2$ (which is clearly semilinear). Thus the first set only adds to $I_2$ a finite number of points. 
        It is not closed however. This is not important in this proof as we do not require this condition here. $\closure{I_2}$ could also be used, but would require some additional arguments.

        \item
        $I_3 = \{ M_p^{j_2} M_d^{j_3} M^{j_4}_-    M_e M^{j_1}_xM_m^iM_w  x \mid i\in \{1,2\}\wedge j_1,j_2,j_3,j_4\in \N\wedge w\in \{1,\dots,9\}^* \wedge |w|\leq n_0+1\}
                \cup  
        \{0\}^{14}\times\{(s^3,c^3,d^3,n^3,r^3,k^3,1)\mid k^3\leq 1 \vee n^3\leq 0 \vee s^3 \geq 0 \vee d^3< 0\}$.\\
        As for $I_2$, the first set represents an infinite number of points, but which all reach the second set after a bounded number of steps.

\end{itemize}

This invariant contains the initial configuration but not the target.
It is semilinear. Let us show that it is stable.

Let $h\in I$, we denote the variables within $h$ of the block of seven dimensions number $i\in \{1,2,3\}$ by
$(s^i,c^i,d^i,n^i,r^i,k^i,a^i)$.

\noindent $\bullet$ Case $h\in I_1$.
If $|s^1| \geq 4(c^1+d^1)$, then for $i\leq 9$, $M_ih$ satisfies the same property and for 
$j\in\{1,2\}$,  $M^j_m h\in \{0\}^{7}\times\{(s^2,c^2,d^2,n^2,r^2,k^2,1)\mid (d^2<0\wedge r^2<0) \vee (s^2 \geq 0 \wedge c^2\geq 0)\}\times\{0\}^{7} \subseteq I_2$. The other matrices do not modify 
$s^1, c^1$ and $d^1$ ensuring the property still holds.
If there exists $w\in \{1,\dots,9\}^*,|w|\leq n_0+1$ such that $h=M_w x$, then for $j\in\{1,2\}$,  $M^j_m h\in I_2$, if $|w|\leq n_0$, for all $i\leq 9$, $M_i h \in I_1$ and if $|w|= n_0+1$, then,as seen
in the proof of Theorem~\ref{th:closed-monniaux}, $|s^1| \geq 4(c^1+d^1)$. 
The other matrices do not modify the variables.

\noindent $\bullet$ Case $h\in I_2$.
Only two matrices can affect $h$ in meaningful ways: $M_x$ and $M_e$.
As $M_e$ transfer the values to the third block without affecting the positivity (or negativity) of
$s^2,d^2$ and $k^2$,  $M_e h\in I_3$.
As $M_x$ removes 1 from $k^2$ and adds $c^2$ to $s^2$ and $r^2$ to $d^2$, if $h$ satisfies $k^2\leq -1 \vee (d^2<0\wedge r^2<0) \vee (s^2 \geq 0 \wedge c^2\geq 0)$, then so does $M_x h$. 
Finally, if $h$ can be written in the form $M^j_xM_m^iM_w x$, then it is immediately also the case for 
$M_x h$.

\noindent $\bullet$ 
Case $h\in I_3$.
The matrices that affect the variables in meaningful ways are $M_p$, $M_d$ and $M_-$.
As they are commutative and respectively increases by 1, decreases by 1 and decreases by 2 the value of $s^3$, $d^3$ and $n^3$, we have that all $M_p h$, $M_d h$ and $M_- h$ are in $I_3$.

Thus $I$ is a separating invariant of the instance.

Conversely, assume the $\omega$-PCP possesses a solution $w\in \{1,\dots,9\}^\omega$. 
For $z\in\nats$, we denote $w\restriction_z$ the prefix of length $z$ of $w$ and $h_z= M_{w\restriction_z}x$. 
Then for all $z\in \nats$, $|s_z^1|<4(c_z^1+d_z^1)$ as showed in the proof of Theorem~\ref{th:closed-monniaux}. 
If $s_z^1\geq 0$ (resp. $s_z^1< 0$), triggering the sequence of matrices $M_e (M_x)^{k^1_z} M^1_m$ (resp. $M_e (M_x)^{k^1_z} M^2_m$) on $h_z$
one reaches a point $g_z \in \{0\}^{14}\times \{(s^3,0,d^3,2z,0,0,1)\}$ where 
$s^3$ is a strictly negative integer and $d^3$ is a positive integer.
Therefore, using the matrices $M_p$ and $M_d$ one can reach the point
$y_{2z}\in \{0\}^{14}\times \{(-1,0,0,2z,0,0,1)\}$.
As this is true for all $z\in \N$, any semilinear invariant $I$ containing all the points $y_{2z}$ contains a point $y_{m}$ with $m\in\N$ an odd integer. Therefore as $I$ is stable under $M_-$, it must contains the target and thus is not a separating invariant of $((M_1,\dots M_9, M^1_m, M^2_m, M_x, M_e, M_-, M_p,M_d),x,y)$.

This concludes the reduction.

\section{Proof of Lemma~\ref{lem:reduc}}\label{appendix:reduc}
Let $\ell = (A,x,y)$ be a non-reach instance in Jordan normal form.

\begin{itemize}
        \item Suppose $A$ is not invertible, we distinguish two cases.
                \begin{itemize}
                        \item If for some Jordan block $J$ associated with the eigenvalue $0$, we have that $y$ is not the zero vector,
                        \emph{i.e.}, $y_J \neq 0$,
                        then $\I = \set{x, Ax, \ldots, A^{d-1}x} \cup \set{z \in \C^d \mid z_J = 0}$ is a semilinear invariant.
                        Indeed, the Jordan block $J$ is nilpotent, so for any point $u$ and $n \ge d$, we have that $J^n u = 0$,
                        so in particular $(A^n x)_J = 0$.
                        Moreover, since by assumption $y$ is not reachable, it is not one of $A^n x$ for $n < d$, 
                        and $y_J \neq 0$, so $y \notin \I$.

                        \item Otherwise, let $J$ be a Jordan block associated with the eigenvalue $0$, such that $y_J = 0$.
                        Consider the Orbit instance $\ell' = (A_{\compl J},(A^d x)_{\compl J},y_{\compl J})$.
                        We claim that $\ell$ admits a semilinear invariant if, and only if, $\ell'$ does.

                        Let $\I$ be a semilinear invariant for $\ell$.
                        Build $\I' = \set{z \in \C^{\compl J} \mid (z,0) \in \I}$.
                        We argue that $\I'$ is a semilinear invariant for~$\ell'$.
                        Indeed, $(A^d x)_{\compl J} \in \I'$ since $A^d x \in \I$
                        and $(A^dx)_J = 0$, because the Jordan block $J$ is nilpotent.
                        The stability of $\I'$ under $A_{\compl J}$ is clear, and $y_{\compl J}$ is not in $\I'$
                        because $y_J = 0$, so if $y_{\compl J}$ would be in $\I'$ this would imply that $y$ is in $\I$.

                        Conversely, let $\I'$ be a semilinear invariant for $\ell'$, 
                        let $\I = \I' \times \C^J$, then $\set{x,Ax,\ldots,A^{d-1}x} \cup \I$ is a semilinear invariant for $\ell$.

                        We reduced the existence of semilinear invariants from $\ell$ to $\ell'$, removing one Jordan block $J$
                        such that $y_J = 0$. Proceeding this way for each such Jordan block,
                        we reduce to the case where the matrix is invertible.
                \end{itemize}

        \item Suppose $A$ contains a Jordan block $J$ such that $x_{J,d(J)} = 0$.  We distinguish two cases.
                \begin{itemize}
                        \item If for some Jordan block $J$ we have $x_{J,d(J)} = 0$ and $y_{J,d(J)} \neq 0$, 
                        then the set $\I = \set{z \in \C^d \mid z_{J,d(J)} = 0}$ is a semilinear invariant for $\ell$.

                        \item Otherwise, let $J$ be a Jordan block such that $x_{J,d(J)} = y_{J,d(J)} = 0$.
                Write $p$ for the dimension $(J,d(J))$.
                        Consider the Orbit instance $\ell_p = (A_{p^c},x_{p^c},y_{p^c})$,
                        we claim that $\ell$ admits a semilinear invariant if, and only if, $\ell_p$ does.

                        Let $\I$ be a semilinear invariant for $\ell$.
                        Build $\I_p = \set{z \in \C^{p^c} \mid (z,0) \in \I}$, then $\I_p$ is a semilinear invariant for $\ell_p$.
                        Conversely, let $\I_p$ be a semilinear invariant for $\ell_p$,
                        let $\I = \set{z \in \C^d \mid z_p = 0 \text{ and } z_{p^c} \in \I_p}$,
                        then $\I$ is a semilinear invariant for $\ell$.
                \end{itemize}
                We reduced the existence of semilinear invariants from $\ell$ to $\ell_p$, removing the last dimension in a Jordan block $J$
                such that $x_{J,d(J)} = 0$. Proceeding this way for each such Jordan block,
                we reduce to the case where there are no such Jordan blocks.
        \item Suppose $A$ has a diagonal Jordan block $J$, that is, $d(J)=1$, with eigenvalue $\lambda$ with $\lambda^n=1$. We set $n$ minimal such that $\lambda^n = 1$ and distinguish two cases
        \begin{itemize}
                \item If for every $k \leq n-1$, $y_J \neq \lambda^k x_J$, then the set $\I = \bigcup_{k=0}^{n-1} \{z \mid z_J = \lambda^k x_J\}$ is a semilinear invariant for $\ell$.
                \item Otherwise, let $k \leq n-1$ be such that $y_J = \lambda^k x_J$. We claim that there exists an invariant for $\ell$ if and only if there exists an invariant for $\ell' = (A_{\compl J}^n, A_{\compl J}^k x_{\compl J}, y_{\compl J})$. Let $\I'$ be an invariant for $\ell'$. For $k' \in \{0, \dots, n-1\},$ we put
                \[
                \I_{k'} = \{z \mid z_J = \lambda^{k+k'} x_J \text{ and } z_{\compl J} \in A^{k'}_{\compl J} \I'\},
                \]
                $\I= \{x, Ax, \dots, A^{k-1} x\} \cup \bigcup_{k' = 0}^{n-1} \I_{k'}$, and prove that the semilinear set $\I$ is an invariant for $\ell$. It is clear that $x \in \I$. Moreover, $y$ does not belong to $\I$: indeed, $y \notin \{x, Ax, \dots, A^{k-1}x\}$ and $y \notin \bigcup_{k' = 1}^{n-1} \I_{k'}$ because $y_J = \lambda^k x_J \neq \lambda^{k+k'} x_J$ for any $k' \in \{1,\dots,n-1\}$ (we assume $x_J \neq 0$ thanks to a previous reduction), and $y \notin \I_0$ since $\I'$ is an invariant for $\ell'$. Finally, $\I$ is stable for $A$ since $A^k x \in \I_0$, $A \I_{k'} = \I_{k'+1}$ if $k < n-1$ and $A \I_{n-1} = A^n \I_0 \subseteq \I_0$ since $\lambda^n=1$ and $A_{\compl J}^n \I' \subseteq \I'$.

                Conversely, let $\I$ be an invariant for $\ell$. We let $\I'$ be the projection on $\compl J $ of $A^k \I \cap \{z \mid z_J = \lambda^k x_J\}$, and claim it is an invariant for $\ell'$. Indeed, quite clearly $A_{\compl J}^k x_{\compl J} \in \I'$ and $\I'$ is stable for $A_{\compl J}^n$. Now, if $y_{\compl J} \in \I'$ then it must be that $y \in \I$, a contradiction.
        \end{itemize}
        \item Let $J$ be a Jordan block of $A$ with eigenvalue $\lambda<1$ and such that $y_J=0$. If there are infinitely many integers $n$ such that $A_{\compl J}^n x_{\compl J} = y_{\compl J},$ then $y \in \closure{\{A^nx, n \in \N\}}$, so there exists no closed invariant for $\ell$. Otherwise, we let $n_0 \in \N$ be such that $y_{\compl J} \notin \{A_{\compl J}^n x_{\compl J}, n \geq n_0\}$, and claim that $\ell$ is equivalent to $\ell' = (A_{\compl J}, A_{\compl J}^{n_0} x_{\compl J}, y_{\compl J})$. Let $\I'$ be a semilinear invariant for $\ell'$. Then $\I = \{x,Ax,\dots, A^{n_0 -1} x\} \cup \{z \mid z_{\compl J} \in \I'\}$ is an invariant for $\ell$. Conversely, let $\I$ be an invariant for $\ell$. Let $\delta = \frac 1 2 d(y,\I) >0,$ where the distance is defined according to the infinity norm on $\C^d$. Using Lemma~\ref{lem:existssmallinv} from section~\ref{appendix:less_one}, we construct a semilinear $P \subseteq \C^J$ which is stable for $A_J,$ contains $(A^n x)_J$ for some $n$, and is included in $B(0,\delta)$. Let $\I'$ be the projection of $\{z \mid z \in \I \text{ and } z_J \in P\}$ on $\compl J$. Then $A_{\compl J}^n x_{\compl J} \in \I'$ and $\I'$ is stable for $A_{\compl J}$. Assume that $y_{\compl J} \in \I'$, that is, there exists $\tilde y \in P$ such that $y_1=(y_{\compl J},\tilde y) \in \I$. Then $d(y,\I) \leq d(y,y_1) = ||\tilde y|| \leq \delta/2$ which is a contradiction, so $y_{\compl J} \notin \I'$. Finally, $\{A_{\compl J}^{n_0} x_{\compl J}, A_{\compl J}^{n_0 +1} x_{\compl J}, \dots, A_{\compl J}^{n - 1} x_{\compl J}\} \cup \I'$ is an invariant for $\ell'$ which concludes the proof.
\end{itemize}

\section{Some eigenvalue has modulus greater than 1}\label{appendix:more_one}

We start with a simple lemma.

\begin{lemma}
\label{lem:sequence_convex_increasing}
Let $\lambda$ be a complex non-real number of modulus greater than $1$ and $x$ be a non-zero complex number.
Then the sequence of polyhedra in $\C$ $\left(\convex{\set{\lambda^i x \mid i \in \set{0,\ldots,n}}}\right)_{n \in \N}$ 
is strictly increasing and its limit is $\C$.
\end{lemma}

\begin{proof}
To see that the sequence is strictly increasing, observe that for all $n$ in $\N$, we have
\[
\convex{\set{\lambda^i x \mid i \in \set{0,\ldots,n}}} \subseteq \closure{B}(0,|\lambda|^n \cdot |x|).
\]
It follows that $\lambda^{n+1} x$ is not in $\convex{\set{\lambda^i x \mid i \in \set{0,\ldots,n}}}$.

We now prove that its limit is $\C$.
We write $x = |x| e^{i \theta}$ and $\lambda = |\lambda| e^{i \alpha}$, with $\theta,\alpha$ in $[0,2\pi)$.
Since $\lambda$ is not a real number, $\alpha$ is not $0$.
Let $n_0$ in $\N$ such that $n_0 \alpha > 2 \pi$. Observe that $0$ is in $\convex{\set{\lambda^i x \mid i \in \set{0,\ldots,n_0}}}$.

We claim that for all $n \in \N$, we have 
\[
B(0,|\lambda|^n \cdot |x|) \subseteq \convex{\set{\lambda^i x \mid i \in \set{0,\ldots,n + n_0}}}.
\]
Let $z = |z| e^{i \beta}$ such that $|z| < |\lambda|^{n} \cdot |x|$.
Let $p$ in $\set{0,\ldots,n_0 - 1}$ such that $\beta$ is in $[\theta + (n+p) \alpha, \theta + (n+p+1) \alpha)$.
Then $z$ is in $\convex{\set{0, \lambda^{n+p} x, \lambda^{n+p+1} x}}$.

The claim follows, since the union of the balls $B(0,|\lambda|^n \cdot |x|)$ for $n\in\N$ is $\C$.
\end{proof}

\begin{theorem}
\label{thm:non_diag_more_one}
Let $\ell = (A,x,y)$ be a normalized Orbit instance in Jordan normal form.
Assume that $\ell$ is a non-reach instance.
If the matrix $A$ has an eigenvalue whose modulus is greater than $1$,
then there exists a semilinear invariant for $\ell$.
\end{theorem}

On an intuitive level first: some coordinate of $(A^n x)_{n \in \N}$ diverges to infinity, 
so eventually gets larger in absolute value than the corresponding coordinate in $y$. 
This allows us to construct an invariant for $\ell$ by taking the first points and then all points having a large coordinate in the diverging dimension.  
For the invariant to be semilinear we consider the complement of the convex envelope of an initial segment of points.

\begin{proof}
Let $J$ be a Jordan block of $A$ with eigenvalue $\lambda$ of modulus $>1$. Since $\ell$ is non-trivial, we have $x_{J,d(J)} \neq 0$.
We distinguish two cases.
\begin{itemize}
	\item Suppose that $\lambda$ is a real number.
	
	For all $n \in \N$, we have $(A^n x)_{J,d(J)} = \lambda^n x_{J,d(J)}$, so it diverges to infinity in modulus.
	It follows that there exists $n_0$ in $\N$ such that $|(A^{n_0} x)_{J,d(J)}| > 2\sqrt{2} \cdot |y_{J,d(J)}|$.
	Let
\[
\I = \set{x, A x,\ldots, A^{n_0-1} x} \cup \set{z \in \C^d \mid |\Re{z_{J,d}}| + |\Im{z_{J,d}}| \geq 2|y_{J,d(J)}|}.
\]
	We argue that $\I$ is a semilinear invariant for $\ell$.
	The non-trivial point is that $\I$ is stable under $J$.
	First, $A^{n_0} x$ is in $\I$ because 
	\[|\Re{(A^{n_0} x)_{J,d(J)}}| + \Im{(A^{n_0} x)_{J,d(J)}} \ge \frac{1}{\sqrt{2}} \cdot |(A^{n_0} x)_{J,d(J)}| > 2 |y_{J,d(J)}|.\]
	Second, let $z \in \C^d$ such that $|\Re{z_{J,d(J)}}| + |\Im{z_{J,d(J)}}| \geq 2|y_d|$, we have that $(A z)_{J,d(J)} = \lambda z_{J,d(J)}$, so 
	$|\Re{(A z)_{J,d(J)}}| + |\Im{(A z)_{J,d(J)}}| = |\lambda| (|\Re{z_{J,d(J)}}| + |\Im{z_{J,d(J)}}|) > 2|y_{J,d(J)}|$,
	implying that $A z$ is in $\I$.
	Note that the previous equality holds because $\lambda$ is a real number.

	\item Suppose that $\lambda$ is not a real number.

	Consider the sequence $\left(\convex{\set{\lambda^i x_{J,d(J)} \mid i \in \set{1,\ldots,n}}}\right)_{n \in \N}$ of polyhedra in $\C$.
	Thanks to Lemma~\ref{lem:sequence_convex_increasing}, this sequence is strictly increasing and its limit is $\C$.
	Let $n_0$ in $\N$ such that $y_{J,d(J)}$ and $x_{J,d(J)}$ are both in the interior of $\convex{\set{\lambda^i x_{J,d(J)} \mid i \in \set{1,\ldots,n_0}}}$.
	Let us denote this convex set by $C$, and let
\[
\I = \set{x, A x,\ldots, A^{n_0} x} \cup \closure{\set{z \in \C^d \mid z_{J,d(J)} \notin C}}.
\]
	We argue that $\I$ is a semilinear invariant for $\ell$.
	The non-trivial point is that $\I$ is stable under $A$.

	We first need to prove that $A^{n_0 + 1} x$ is in $\I$.
	We have $(A^{n_0 + 1} x)_{J,d(J)} = \lambda^{n_0 + 1} x_{J,d(J)}$, which is not in $\closure{\convex{\set{\lambda^i x_{J,d(J)} \mid i \in \set{1,\ldots,n_0}}}}$, 
	because this sequence of polyhedra is strictly increasing.
	Thus $A^{n_0 + 1} x$ is in $\I$.

	Finally, let $z \in \C^d$ such that $z_{J,d(J)} \notin \closure C$, we show that $A z$ is in $\I$.
	We have $(A z)_{J,d(J)} = \lambda z_{J,d(J)}$.
	Assume towards contradiction that $(A z)_{J,d(J)}$ is in $C$, so $\lambda z_{J,d(J)}$ is a convex combination of 
	$\set{\lambda^i x_{J,d(J)} \mid i \in \set{1,\ldots,n_0}}$.
	This implies that $z_{J,d(J)}$ is a convex combination of $\set{\lambda^{i-1} x_{J,d(J)} \mid i \in \set{1,\ldots,n_0}}$.
	Since $x_{J,d(J)}$ is in $C$, this implies that $z_{J,d(J)}$ is in $C$, which is a contradiction.
	Thus $A z$ is in $\I$, and $\I$ is a semilinear invariant for $\ell$.
\end{itemize}
\end{proof}

\section{Some eigenvalue has modulus less than 1}\label{appendix:less_one}

We start with a simple lemma.

\begin{lemma}
\label{lem:sequence_convex_decreasing}
Let $\lambda$ be a complex non-real number of modulus less than $1$ and $x$ be a non-zero complex number.
Then the sequence $\left(\convex{\set{\lambda^i x \mid i \in \set{0,\ldots,n}}}\right)_{n \in \N}$ of polyhedra in $\C$ is ultimately constant,
and its limit contains an open neighbourhood of $0$.
\end{lemma}

\begin{proof}
We write $x = |x| e^{i \theta}$ and $\lambda = |\lambda| e^{i \alpha}$, with $\theta,\alpha$ in $[0,2\pi)$.
Since $\lambda$ is not a real number, $\alpha$ is not $0$.
Let $n_0$ in $\N$ such that $n_0 \alpha > 2 \pi$. Observe that $0$ is in $\convex{\set{\lambda^i x \mid i \in \set{0,\ldots,n_0}}}$.

We claim that $B(0,|\lambda|^{n_0} \cdot |x|)$ is included in $\convex{\set{\lambda^i x \mid i \in \set{0,\ldots,n_0}}}$.
Let $z = |z| e^{i \beta}$ such that $|z| < |\lambda|^{n_0} \cdot |x|$.
Let $p$ in $\set{0,\ldots,n_0 - 1}$ such that $\beta$ is in $[\theta + p \alpha, \theta + (p+1) \alpha)$.
Then $z$ is in $\convex{\set{0, \lambda^p x, \lambda^{p+1} x}}$.

The claim follows, since for all $n > n_0$, we have that $\lambda^n x$ is in $B(0,|\lambda|^{n_0} \cdot |x|)$.
\end{proof}

The following Lemma is the cornerstone for this section.

\begin{lemma}\label{lem:existssmallinv}
Let $\varepsilon >0$ and $\lambda \in \C$ with $|\lambda|<1$. There exists a semilinear set $\I \subseteq B(0, \varepsilon) \subseteq \C^d$ which is stable for $\Jord{d}{\lambda}$, and contains $B(0,\varepsilon')$ for some $0<\varepsilon'<\varepsilon$.
\end{lemma}

\begin{proof}
We let $J$ denote $\Jord d \lambda$, and prove the Lemma by induction on $d$. We first treat the case where $\lambda \in \R$. Let
\[
\I = \left\{z \in \C^d \mid \forall i, |\Re{z_i}| + |\Im{z_i}| \leq \varepsilon (1 - |\lambda|)^i \right\} \subseteq B(0, \varepsilon).
\]
Then $B(0, \varepsilon (1 - |\lambda|)^d /2) \subseteq \I$. We show that $J\I \subseteq \I$. Let $z \in \I$. Then $(Jz)_d = \lambda z_d$, so $|\Re{(Jz)_d}|+|\Im{(Jz)_d}| \leq |\lambda|(|\Re{z_d}| + |\Im{z_d}|) \leq \varepsilon (1 - |\lambda|)^d$. Now if $i<d$, $(Jz)_i = \lambda z_i + z_{i+1}$, so
\[
\begin{aligned}
|\Re{(Jz)_i}|+|\Im{(Jz)_i}| &= |\lambda \Re{z_i} + \Re{z_{i+1}}| + |\lambda \Im{z_i} + \Im{z_{i+1}}| \\
& \leq |\lambda|(|\Re{z_i}| + |\Im{z_i}|) + (|\Re{z_{i+1}}| + |\Im{z_{i+1}}|) \\
& \leq |\lambda| \varepsilon (1-|\lambda|)^i +  \varepsilon (1-|\lambda|)^{i+1} = \varepsilon (1-|\lambda|)^{i}.
\end{aligned}
\]
Hence $\I$ is stable for $J$, which concludes this first case. We now assume that $\lambda \notin \R$. We start with the base case $d=1$. Let $u \in \C$ of modulus $\varepsilon$, for instance $u = \varepsilon$. We let $\I =\convex{\set{\lambda^i u \mid i \in \set{0,\ldots,p}}}$. Since extremal points of $\I$ are of the form $\lambda^i u$, of modulus $|\lambda|^i \varepsilon < \varepsilon$, it holds that $\I \subseteq B(0,\varepsilon)$.

Let $d>1$, and assume the result known for smaller dimensions. We let $u$ be a complex number of modulus $\varepsilon /2$, for instance, $u=\varepsilon/2 \in \C$ . We let $\alpha= |\lambda|^p \varepsilon$ which is such that $\alpha < \varepsilon$ and $B(0, \alpha) \subseteq \convex{\set{\lambda^i u \mid i \in \set{0,\ldots,p}}}$ and put $\varepsilon' = \frac \alpha 2 (1 - |\lambda|)$. We let $\I'$ be a semilinear subset of $\C^{d-1}$ given by induction, stable for $\Jord{d-1}{\lambda}$, and such that
\[
B(0, \varepsilon'') \subseteq \I' \subseteq B(0, \varepsilon') \subseteq \C^{d-1},
\]
for some $\varepsilon''>0$. In particular, $0 \in \pi_1(\I')$.

We consider the sequence $(C_j)_{j \in \N}$ of semilinear subsets of $\C$ given by $C_0 = \{u\}$ and for all $j$,
\[
C_{j+1} = \{\lambda z + z', z \in C_j, z' \in \pi_1(\I')\}.
\]
Let us know prove two facts about the sequence $(C_j)_j$.
\begin{itemize}
\item For all $j$, and $z \in C_j$, $|z| \leq |\lambda|^j \frac \varepsilon 2 + \varepsilon' \sum_{i=0}^{j-1} |\lambda|^i$, which we prove by induction. This is clear for $j=0$, and if it holds for elements $z \in C_j$, an element $\lambda z + z' \in C_{j+1}$ with $z' \in \pi_1(\I')$ is such that
\[
|\lambda z + z'| \leq |\lambda|\left(|\lambda|^j \frac \varepsilon 2 + \varepsilon'\sum_{i=0}^{j-1} |\lambda|^i \right) + \varepsilon' \leq |\lambda|^{j+1} \frac \varepsilon 2 + \varepsilon' \sum_{i=0}^{j} |\lambda|^i.
\]
\item There exists $j_0$ such that $C_{j_0} \subseteq B(0, \alpha) \subseteq \convex{\{C_0, \dots C_{j_0-1}\}}$. Indeed, the sequence $|\lambda|^j \frac \varepsilon 2 + \varepsilon' \sum_{i=0}^{j-1} |\lambda|^i$ goes to $\frac{\varepsilon'}{1 - |\lambda|}$, so for large enough $j$, $C_j \subseteq B(0, \frac{2\varepsilon'}{1 - |\lambda|}) = B(0, \alpha)$. Now $0 \in \pi_1(\I')$, so by an easy induction, $\lambda^j u \in C_j$. Hence,
\[
C_{j_0} \subseteq B(0, \alpha) \subseteq \convex{\{\lambda^j u, j \in \N\}} \subseteq \convex{\{C_j, j \in \N\}} \subseteq \convex{\{C_0, \dots, C_{{j_0}-1}\}}.
\]
\end{itemize} 
We now let
\[
\I = \convex{\{C_0, \dots, C_{j_0-1}\}} \times \I',
\]

Then
\begin{itemize}
\item $\I \subseteq B(0, \frac \varepsilon 2 + \frac {\varepsilon'}{1 - |\lambda|}) \subseteq B(0, \frac \varepsilon 2 + \frac \alpha 2) \subseteq B(0, \varepsilon)$.
\item $B(0, \min(\alpha, \varepsilon'')) \subseteq \I$, and
\item $\I$ is stable for $J$, because $J(C_j \times \I') \subseteq C_{j+1} \times \I'$.
\end{itemize}
This concludes the induction, and the proof of the Lemma.
\end{proof}

We may now prove the following Theorem.

\begin{theorem}\label{thm:less_one}
Let $\ell = (x,A,y)$ be a normalized Orbit instance. Assume that $\ell$ is a non-reach instance. If the matrix $A$ has an eigenvalue whose modulus is smaller than 1, then there exists a semilinear invariant for $\ell$.
\end{theorem}

\begin{proof}
Let $J$ be a Jordan block of $A$ with eigenvalue $\lambda$ such that $|\lambda|<1.$ Since $\ell$ is normalized, $y_J \neq 0$. Let $\varepsilon = |y_J|/2$. Using Lemma~\ref{lem:existssmallinv}, we obtain $\varepsilon'>0$ and a semilinear set $\I \subseteq \C^{d(J)}$ such that $J\I \subseteq \I$ and $B(0, \varepsilon') \subseteq \I \subseteq B(0, \varepsilon)$. Now $(A^nx)_J \to 0$, so there exists $n_0$ such that $(A^{n_0}x)_J \in B(0, \varepsilon') \subseteq \I$. Hence, it is easy to see that
\[
\{x, Ax, \dots, A^{n_0-1}x\} \cup \{z \in \C^d \mid z_J \in \I\}
\]
is a semilinear invariant for $\ell$.

\end{proof}

\section{Some non-diagonalisable eigenvalue is a root of unity}\label{appendix:mod_one_rou_not_diag}

\begin{theorem}
Let $\ell = (A,x,y)$ be a normalized Orbit instance.
Assume that $\ell$ is a non-reach instance.
If $A$ contains a non-diagonal Jordan block $J$ whose eigenvalue is a root of unity,
then there exists a semilinear invariant for $\ell$.
\end{theorem}

\begin{proof}
Let $J$ be a non-diagonal Jordan block of $A$ with eigenvalue $\lambda$ with $\lambda^m=1$. We shall use divergence on the coordinate $(J, d(J)-1)$ to construct an invariant. Recall that $x_{J,d(J)} \neq 0$. For any $n \in \N$, we have $(A^nx)_{J,d(J)-1} = \lambda^n x_{J,d(J)-1} + n \lambda^{n-1}x_{J, d(J)}$, and $(A^nx)_{J, d(J)}= \lambda^n x_{J,d(J)}$. Hence,
\[
\Re{\lambda(A^n x)_{J,d(J)-1} \overline{(A^nx)_{J, d(J)}}} = \Re{\lambda x_{J, d(J) -1} x_{J, d(J)}} + n |x_{J,d(J)}|^2,
\]
which goes to infinity when $n$ grows. Note that this condition is quadratic, but since $(A^nx)_{J, d(J)}$ takes only a finite number of values, we will be able to state it in a semilinear fashion. Let $n_0$ be such that $M = \Re{\lambda(A^{n_0} x)_{J,d(J)-1} \overline{(A^{n_0}x)_{J, d(J)}}} > \Re{\lambda y_{J,{d(J)-1}} \overline{y_{J,d(J)}}}$. Let
\[
\I = \{x, Ax, \dots, A^{n_0 - 1} x\} \cup \bigcup_{i=0}^{m-1} \I_i,
\]
where
\[
\I_i = \{z \in \C^d \mid z_{J,d(J)} = \lambda^i x_{J,d(J)} \text{ and } \Re{\lambda z_{J,d(J)-1} \overline{z_{J,d(J)}}} \geq M)\}.
\]
It is clear that $x \in \I$ and $y \notin \I$. Each $\I_i$ is semilinear because the second condition is actually semilinear assuming $z_{J,d(J)} = \lambda^i x_{J,d(J)}$. Now if $z \in \I_i$, we obtain that $(Az)_{J,d(J)} = \lambda z_{J,d(J)} = \lambda^{i+1} x_{J,d(J)}$, and
\[
\Re{\lambda (Az)_{J,d(J)-1} \overline{Az_{J,d(J)}}} = \Re{\lambda z_{J,d(J)-1} \overline{z_{J,d(J)}}} + |z_{J,d(J)}|^2 \geq M
\]
so $Az \in \I_{i+1}$ if $i<m$, and $Az \in \I_0$ (since $\lambda^m=1)$ if $i=m$. Hence $\I$ is stable for $A$.
\end{proof}

\section{All eigenvalues have modulus 1 and are not roots of unity}\label{appendix:mod_one_not_rou}

We finally deal with the most involved case, namely, when all eigenvalues have modulus 1 and none are roots of unity. In this setting, we will be able to describe the \emph{minimal semilinear inviariant} for $A$ and $x$, that is, a semilinear invariant which is contained in any semilinear invariant. We say that two eigenvalues are \emph{equivalent} if their quotient is a root of unity. Intuitively, the only nontrivial semilinear relations that invariants will be able to exploit are the ones that hold among equivalent blocks.

We now give a high-level overview for this section.
\begin{itemize}
\item We first assume that eigenvalues are pairwise non-equivalent. The aim is to show that in this setting, any semilinear invariant is trivial. This is the object of subsection~\ref{subsec:mod_one_not_rou_not_eq}
\begin{itemize}
\item It is first shown how to deal with the diagonal case. This makes a crucial use of the Skolem-Mahler-Lech Theorem.
\item We then extend to general (possibly non-diagonal) blocks by induction on the total dimension. This is the most technical part of the proof.
\end{itemize}
\item We then deal with equivalent eigenvalues in subsection \ref{subsec:mod_one_not_rou}: we first treat the case where all equivalent eigenvalues are equal, and then show how to easily reduce to this setting.
\end{itemize}

\subsection{All eigenvalues are non-equivalent}\label{subsec:mod_one_not_rou_not_eq}
\subsection*{The diagonal case}\label{subsec:diag}

We will make use of the following powerful theorem about linear recurrence sequences.
This result is due to Skolem~\cite{Sko34}, and more general versions were subsequently obtained by
Mahler~\cite{Mah35,Mah56} and Lech~\cite{Lec53}.

\begin{theorem}[Skolem, Mahler, Lech]
\label{thm:SML}
Let $(u_n)_{n \in \N}$ be a real non-degenerate linear recurrence sequence, that is, $u_n = \sum_{i=1}^d v_i \lambda_i^n$, for some $v \in \C^d \setminus \{0\}$, where for any $i \neq j$, $\frac{\lambda_i}{\lambda_j} \notin \U$.
Then $\set{n \in \N \mid u_n = 0}$ is finite.
\end{theorem}

We write $A = \diag(\lambda_1,\ldots,\lambda_d)$ for 
\[
A =
\begin{bmatrix}
\lambda_1 \\ & \lambda _2 \\ & & \ddots \\ &&& \lambda_d
\end{bmatrix}.
\]

\begin{lemma}
\label{lem:all_mod_1_diag}
Let $\lambda_1,\ldots,\lambda_d \in S^1$ and $A = \diag(\lambda_1,\ldots,\lambda_d)$.
Assume that:
\begin{itemize}
	\item for all $i$, we have $\lambda_i \notin \U$, and
	\item for all $i,j$ such that $i \neq j$, we have $\frac{\lambda_i}{\lambda_j} \notin \U$.	
\end{itemize} 
Let $\I$ be a non-empty closed semilinear set stable under $A$, which moreover containes a point $x \in \I$ such that for all coordinate $i$, $x_i \neq 0$. Then $\I = \C^d$.
\end{lemma}

\begin{proof}
Let $\I$ be such a semilinear set, we show a few facts:
\begin{itemize}
	\item[$(i)$] $\I$ must have even dimension over $\R$,
	\item[$(ii)$] $\I$ must have dimension $>2d-2$ over $\R$ (hence, $\I$ has full dimension thanks to $(i)$),
	\item[$(iii)$] $\partial \I$ is stable under $A$,
	\item[$(iv)$] if it is non-empty (that is, if $\I \neq \C^d$), $\partial \I$ contains a point which is nonzero on each coordinate.
\end{itemize} 
This implies the desired result: if towards contradiction we would have that $\I \neq \C^d$, 
then $\partial \I$ would be a non-empty closed semilinear set stable under $A$ thanks to $(iii)$, it would contain a point which is nonzero on each coordinate thanks to $(iv)$, but yet it cannot have full dimension. We now prove the four claims.

\begin{itemize}
\item[$(i)$] 

Let $s = \dim_\R(\I)$. Then $\I$ is contained into the union of finitely many affine subspaces of dimension $s$, write
\[
\I \subseteq \bigcup_{i=1}^p F_i,
\]
where $F_i \subseteq \C^{d} \simeq \R^{2d}$ is a real affine space of dimension $s$, of direction $F_i - F_i = E_i \subseteq \R^{2d}$. We first show that for some $i$, $E_i$ must be stable for some power of $A$ (seen as a transformation of $\R^{2d}$), and then that this implies that $s$ is even.

Since $\dim_\R(\I) \geq s$, there must be $\tilde x \in \I$ and $\varepsilon >0$ such that
\[
B(\tilde x, \varepsilon) \cap \I = B(\tilde x, \varepsilon) \cap F_i
\]
for some $i$. Then for all $n$, $A^n (B(\tilde x, \varepsilon) \cap F_i )= B(A^n \tilde x, \varepsilon) \cap A^n F_i \subseteq \I$, and has dimension $s$ over $\R$, hence there exists $i_n$ such that $B(A^n \tilde x, \varepsilon) \cap \I = B(A^n \tilde x, \varepsilon) \cap F_{i_n}$. Now let $n_1 < n_2$ be such that $i_{n_1} = i_{n_2}=i,$ let $n = n_2 - n_1$ and let $x = A^{n_1} \tilde x$. We show that $E_i$ is stable under $A^n$.

Let $e \in E_i$, and let $\tilde e=\varepsilon \frac e {2||e||}$. Then $x + \tilde e \in B(x, \varepsilon) \cap F_i \subseteq \I$ so $A^n(x + \tilde e) = A^n x + A^n \tilde e \in B(A^n x, \varepsilon) \cap \I \subseteq F_i$ so $A^n \tilde e = A^n x + A^n \tilde e - A^n x \in F_i - F_i = E_i$ and since $E_i$ is $\R$-linear, $A^n e \in E_i$.

Now since the $\lambda_i$'s are not roots of $1$, $A^n$, when seen as a transformation of $R^{2d}$, is the product of $d$ diagonal irrationnal rotations of $\R^2$. Such a map only stabilizes linear spaces of even dimensions, hence $s = \dim_\R(E_i)$ is even.

\vskip1em
\item[$(ii)$] Assume for contradiction that $\dim_\R( \I) \leq 2d-2$. Let $x \in \I$ be a point with $x_i \neq 0$ for all $i\in \{1,\dots, d\}$. Now $\I \subseteq \bigcup_{i=1}^p F_i$, where the $F_i$'s are affine spaces of real dimension $2d-2$, that is, spaces of the form
\[
F_i = \{z \in \C^d \mid \sum_{i=1}^d u_i z_i = a\},
\]
for some nonzero $u \in \C^d$ and some $a \in \C$. Consider the orbit $\O = \{A^n x, n \in \N\}$ of $x$. There must be $i$ such that $F_i \cap \O$ is infinite, hence there are infinitely many $n \in \N$ such that
\[
\sum_{i=1}^d \lambda_i^n u_i x_i = a,
\]
which contradicts Theorem~\ref{thm:SML} applied to eigenvalues $\lambda_1,\dots, \lambda_d,1$, since $u \neq 0$ implies $(u_i x_i)_{i \in \{1,\dots, d\}} \neq 0$.

\vskip1em
\item{$(iii)$}
We argue that $\C^d \setminus \I$ is stable under $A$,
which together with the fact that $\I$ is stable under $A$ implies that $\partial \I$ is stable under $A$.
Equivalently we show that $\I$ is stable under $A^{-1}$: let $x$ in $\I$, we prove that $A^{-1} x$ is in $\I$.

Let
\[ 
L_A = \set{v \in \Z^d \mid \lambda_1^{v_1} \cdots \lambda_d^{v_d} = 1}
\] 
be the set of all multiplicative relations holding among $\lambda_1,\ldots,\lambda_d$.  
Notice that $L_A$ is an additive subgroup of $\Z^d$.
Consider the set of diagonal $d\times d$ matrices
\[ 
{T}_A = \set{ \diag(\mu_1,\ldots,\mu_d) \mid \mu \in S^d \mbox{ and } \forall v \in L_A \, ( \mu_1^{v_1} \cdots \mu_d^{v_d} = 1)} 
\] 
whose diagonal entries satisfy the multiplicative relations in $L_A$.
Using Kronecker's Theorem on inhomogeneous simultaneous Diophantine
approximation~\cite{C65}, it is shown in~\cite[Proposition 3.5]{OW14}
that $\set{ A^n : n \in \N}$ is a dense subset of ${T}_A$.
This immediately gives
\[
\overline{ \set{A^nx \mid n \in \N} } = \set{Mx \mid M \in {T}_A} \, .
\]
Since $x$ is in $\I$ and $\I$ is stable under $A$, we have that $\overline{ \set{A^nx \mid n \in \N} } \subseteq \overline{\I} = \I$.
Observe furthermore that $A^{-1} = \diag(\lambda_1^{-1},\ldots,\lambda_d^{-1})$ is in ${T}_A$,
so thanks to the previous equality $A^{-1} x$ is in $\I$.

\vskip1em
\item{$(iv)$} Let $\F = \bigcup_{i=1}^d \C^{i-1} \times \{0\} \times \C^{d-i}$ be the set of points with at least one zero coordinate. Assume for contradiction that $\partial \I \subseteq \F$. Let $x \in \I \setminus \F$ and $y \in \compl{\I} \setminus \F$, which is non-empty because $\compl{\I}$ is a nonempty open subset of $\C^d$ whereas $\F$ has empty interior. Now $\C^d \setminus \F$ is path connected, so there exists a path from $x$ to $y$ which avoids $\F \supseteq \partial \I$, a contradiction.

\end{itemize}
\end{proof}

Although $\forall i,x_i \neq 0 $ holds in a normalized instance, we shall need a slightly stronger result which is a consequence of the previous Lemma.

\begin{theorem}\label{thm:all_mod_1_diag}
Let $A = \diag(\lambda_1, \dots \lambda_d)$ with $\lambda_i \notin \U$ and for $i \neq j, \frac {\lambda_i}{\lambda_j} \notin \U$. Let $\I \subseteq \C^d$ be a closed semilinear set such that $A\I \subseteq \I$. Then $\I$ is a union of sets of the form
\[
\prod_{i=1}^d \varepsilon_i,
\]
where $\varepsilon_i \in \{\{0\}, \C\}$.
\end{theorem}

\begin{proof}
We show that for any $x \in \I$, $\I$ must contain $\prod_i \varepsilon_i$, with $\varepsilon_i = \begin{cases} \{0\} \text{ if } x_i=0 \\ \C \text{ otherwise} \end{cases}$, which implies the wanted result. This is an easy application of the previous Lemma to the projection of $\I \cap \prod_i \varepsilon_i$ on coordinates $\{i \mid x_i \neq 0\}$.
\end{proof}

\subsection*{General case}

We now work with a general (not necessarily diagonal) matrix $A$, whose eignvenvalues are not roots of unity and pairwise non equivalent. The following theorem is proved by induction on $d=\sum d_i$. We write $\J$ for the set of jordan blocks of $A$, and $s = |\J|$.

\begin{theorem}\label{thm:all_mod_1_not_equiv}
Semilinear invariants for $A$ are unions of sets of the form
\[
\prod_{J  \in \J} \C^{p_J} \times \{0\}^{d(J) - p_J},
\]
where for each $J$, $p_J$ is an integer in $\{0,\dots, d(J)\}$.
\end{theorem}
Recall that if $S \subseteq \{(J,i), J \in \J, i \leq d(J)\}$ is a subset of dimensions, $\pi_S : \C^d \to \C^S$ denotes the projection on the coordinates in $S$. We let $\last = \{(J,d(J)), J \in \J\}$ be the set of last coordinates of each block, and for each $J$, we let $P_J = \pi_{\{(J,d(J))\}}^{-1}(\{0\}) \subseteq \C^d$. Note that any set of the form $\prod_J \C^{p_J} \times \{0\}^{d(J) - p_J}$ which is not $\C^d$, is included in $\cup_J P_J$.

The case $d=1$ is proved in section~\ref{subsec:diag}.

We start the induction with an intermediate result.

\begin{lemma}\label{lem:inter}
A semilinear invariant for $A$ is either a union of sets of the form $\prod_{J  \in \J} \C^{p_J} \times \{0\}^{d(J) - p_J}$, or contains $\pi_{last}^{-1}(\{0\}) = \prod_J \C^{d(J)-1} \times \{0\} = \F.$
\end{lemma}

\begin{proof}
Let $\I$ be a semilinear invariant for $A$. Consider $\I'=\pi_\last(\I) \subseteq \C^s$. If $S' \subseteq \J$, we let $\pi'_{S'} : \C^s \to \C^{S'}$ denote the projection on the coordinates of $\C^s$ corresponding to last coordinates of blocks from $S'$. Just like previously, let $P'_J = {\pi'}_{J}^{-1}(\{0\}) \subseteq \C^s$. Since $\I'$ is stable for $\diag(\lambda_1, \dots,\lambda_s)$, it must be that $\I'$ is either $\C^s$, or $\I' \subseteq \cup P'_J$, by Theorem~\ref{thm:all_mod_1_diag}. We reduce to the former case.

Indeed, if $\I' \subseteq \cup P'_J$, we let $\I_J = \pi_{\last}^{-1}(P'_J)$, so that $\I = \cup \I_J.$ Now, $\I_J \subseteq P_J$, so $\pi_{\compl{\{(J, d(J))\}}}(\I_J)$ is stable for the matrix $A'$ obtained from $A$ just by diminishing the dimension of block $J$ by 1. By induction, $\pi_{\compl{\{(J, d(J))\}}}(\I_J)$ is a union of sets of the form $\prod_{J' \in \J'} \C^{p_{J'}} \times \{0\}^{d(J') - p_{J'}}$ (where $\J'$ is the set of Jordan blocks of $A'$), so $\I_J$ has the wanted form, and so does $\I = \cup I_J$.

Hence we assume that $\I'=\C^s.$ We aim to show that $\I \supseteq \F$, or equivalentely, $\pi_{\compl{\last}}(\I \cap \F) = \C^{d-s}.$ By induction, since it is stable for the matrix obtained from $A$ by diminishing the dimension of each block by 1, we know that either $\pi_{\compl{\last}}(\I \cap \F) = \C^{d-s},$ or $\pi_{\compl{\last}}(\I \cap \F) \subseteq \cup_J \pi_{\{(J,d(J)-1)\}}^{-1}(\{0\}).$ We assume the latter towards contradiction. In plain English, any $z \in \I$ that is $0$ on the last coordinate of each block (that is, $z \in \I \cap \F$) must have one of its prior coordinates (that is, $(J, d(J)-1)$ for some $J$) which is zero. We will now project on only the last two coordinates of each block. Formally, we let 
\[
\lasttwo = \bigcup\limits_{\substack{J \in \J \\ d(J) \geq 2}} \{(J, d(J)-1), (J, d_J)\} \cup \bigcup\limits_{\substack{J \in \J \\ d(J) = 1}} \{(J, 1)\},
\]
and consider $\I''= \pi_{\lasttwo} (\I)$.
Then $\I''$ is stable for $A''$, the matrix obtained from $A$ by reducing the size of each block of size $\geq 3$ to 2. We let $\J''$ denote the set of Jordan blocks of $A''$. In particular, any $J'' \in \J''$ is such that $d(J'') \in \{1,2\}$.

For each $J'' \in \J''$, we let $(z_{J'',d(J'')}^{(n)})_{n \in \N}$ be a decreasing sequence of complex numbers that goes to $0$, and such that for any given $n$, the moduli of the $z_{J'',d(J'')}^{(n)}$ are all equal. Since $\I'=\C^s$, for all $n\in \N$ and each $J''$ such that $d(J'')=2$ there exists $z_{J'',1}^{(n)}$ such that $z^{(n)} \in \I''$. By Lemma~\ref{lem:section}, we may pick these values such that $z^{(n)}$ is bounded. Up to extracing a subsequence, we assume without loss of generality that $z^{(n)}$ converges, to, say, $z \in \I''$. Since $z_{J'',d(J'')} = 0$ for all $J'' \in \J''$, $z \in \pi_\lasttwo(\I \cap \F)$, so there must exist $J''_0$ with $d(J''_0)=2$ such that $z_{J''_0, 1} = 0$. We let $\lambda_0$ be the eigenvalue of block $J''_0$. We put $\delta = \min \left(1, \min_{\{J'' \mid d(J'')=2 \text{ and } z_{J'', 1} \neq 0\}}\{|z_{J'',1}|\}\right) >0$. We let $n$ be large enough so $||z^{(n)} - z|| \leq \delta/4 $.  Consider $({A''}^k z^{(n)})_{J''_0,1} = \lambda_0^n (z_{J''_0, 1}^{(n)} + k \lambda_0^{-1} z_{J''_0, 2}^{(n)})$. Let $k(n)=\left \lceil \frac \delta {2|z_{J''_0, 2}^{(n)}|} \right \rceil$. Note that $k(n)$ does not depend on the choice of $J''_0$ since the $z_{J'',d(J'')}^{(n)}$ all have the same moduli. Then
\[
\begin{aligned}
\delta /4=\delta /2 - \delta/4 \leq k(n) |z_{J''_0, 2}^{(n)}| - |z_{J''_0, 1}^{(n)}| \leq |({A''}^{k(n)} z^{(n)})_{J''_0,1}| \\ \leq |z_{J''_0, 1}^{(n)}| + k(n) |z_{J''_0, 2}^{(n)}| \leq \delta/4 + (\frac{\delta}{2|z_{J''_0, 2}^{(n)}|} + 1)|z_{J''_0,2}^{(n)}| \leq \delta.
\end{aligned}
\]
Likewise, we may bound away from zero (which is the reason motivating the choice of $\delta$), and also from above, the moduli of $({A''}^{k(n)} z^{(n)})_{J''_1, 1}$ when $J''_1$ is such that $z_{J''_1,1}\neq 0$. More precisely, 
\[
\begin{aligned}
\delta/4 = \delta - (\delta/4 + \delta/2) \leq |z_{J''_1,1}| - |z_{J''_1,1}^{(n)} - z_{J''_1,1} + k(n) \lambda^{-1} z_{J''_1,2}^{(n)}| \leq |({A''}^{k(n)}z)_{J''_1,1}| \\ \leq |z_{J''_1,1}^{(n)} - z_{J''_1,1}| + |z_{J''_1,1}| + k(n)|z_{J''_1,2}^{(n)}| \leq  \delta /4 + |z_{J''_1,1}| + \delta /2 + \delta /4 \leq |z_{J''_1,1}| + \delta.
\end{aligned}
\]
Now, $\I''$ being stable for $A''$, the sequence $({A''}^{k(n)}z^{n})_n$ has its elements in $\I''$, and ultimately lies in the compact
\[
K = \{u \mid \forall J'', u_{J'',d(J'')} \leq \delta/4 \text{ and } \forall J'' \text{ such that } d(J'')=2, \delta /4 \leq |u_{J'',1}| \leq |z_{J'',1}| + \delta \}.
\]

We may then extract a converging subsequence in $K$, with its limit in $\I''$ such that the last coordinate of each block is zero whereas the previous one is nonzero, a contradiction. This concludes the proof of the Lemma.
\end{proof}

With Lemma~\ref{lem:inter} in hands, we now move on to the proof of Theorem~\ref{thm:all_mod_1_not_equiv}.

\begin{proof}
Let us assume for contradiction that $\I$ is not in the form of the statement of the Theorem. In particular, $\I \neq \C^s$. By Lemma~\ref{lem:inter}, $\F \subseteq \I.$ The set $\I$ is a finite union of closed polyhedra, write $\I = \bigcup_{\P \in P} \P.$ Each polyhedron $\P$ is a finite intersection of closed half-spaces, write $\P = \bigcap_{\H \in H_\P} \H.$ Let $H = \bigcup_{\P \in P} \H_\P$.

We start with a restriction: we may restrict to the case where there is a polyhedron $\P_0 \in P$ such that $\P_0 \cap \F$ has dimension (over $\R$) $2(d-s)$, and $\P_0$ is not included in $\F$.

Let $P_{full}$ denote the set of polyhedra $\P$ in $P$ such that $\P \cap \F$ has dimension $2(d-s)$. Since $\F \subseteq \I,$ $P_{full} \neq \emptyset$. Assume that for each polyhedron $\P$ of $P_{full}$ we have $\P \subseteq \F$. Let
\[
\I' = \closure{\I \setminus \F} \subseteq \bigcup_{\P \notin P_{full}} \P.
\]
Since $\I' \cap \F$ has dimension at most $2(d-s)-1 < \dim_\R(\F)$, it may not be the case that $\F \subseteq \I'$. Now $\compl \F$ is stable for $A$, so so is $\I'$. Hence $\I'$ is, by Lemma~\ref{lem:inter} in the form of the Theorem. Finally, $I = \Q \cup \I''$ is in the wanted form. Hence, we now assume the existence of $\P_0 \in P_{full}$ which is not contained in $\F$.

Let $H_{general}$ be the family of half-spaces in $H$ which are not of the form $\pi_\last^{-1}(\H')$ where $\H'$ is a half-space of $\C^s$ with $0 \in \partial \H'$. Equivalentely, half-spaces in $H_{general}$ are those which do not contain $\F$ in their border. Now if $\H \in H_{general}$ then $\partial \H \cap \F$ has dimension $<2(d-s)$.

It follows that the countable union $\bigcup_{\H \in H_{general}} \bigcup_{k \in \N} A^{-k} \partial \H \cap \F$ has dimension $<2(d-s)$, so it may not cover $\P_0 \cap \F$. Let $z \in \P_0 \cap \F$ be out of this union.

Let $k$ in $\N$. We choose $\varepsilon_k>0$ such that for each $\H \in H$, the set $B(A^kz, \varepsilon_k) \cap \H$ is either empty, the whole ball $B(A^kz, \varepsilon_k)$, or a half-ball of the form $B(A^kz, \varepsilon_k) \cap \pi_\last^{-1}(\H')$, where $\H'$ is a half-space of $\C^s$ such that $0 \in \partial \H'.$ This is achieved by the following case disctinction:
\begin{itemize}
\item Either $A^kz$ is in $\interior{\H},$ then there exists $\varepsilon_k > 0$ such that $B(A^k z, \varepsilon_k) \cap \H = B(A^k z, \varepsilon_k).$
\item Or $A^kz$ is in $\partial \H$. Recall that by construction $A^kz$ is not in $\partial \H$ for $\H$ in $H_{general}$, so $\H \notin H_{general}$ and we are in the third case.
\item Or $A^kz$ is not in $\H$, in which case there exsits $\varepsilon_k>0$ such that $B(A^kz, \varepsilon_k) \cap \H = \varnothing.$ 
\end{itemize}
Without loss of generality, we pick $(\varepsilon_k)_k$ to be decreasing.

It follows that for a polyhedron $\P \in P$, its trace on $B(A^kz, \varepsilon_k)$ is either empty or of the form
\[
B(A^kz, \varepsilon_k) \cap \P = B(A^kz, \varepsilon_k) \cap \pi_{last}^{-1}\left(\bigcap_{\H \in H_\P^{k}} \H \right)
\]
where $H_\P^{k}$ is a finite set of closed half-spaces $\H$ of $\C^s$ such that $0 \in \partial \H$. For $\P \in P$, let $C_{\P,k} = \bigcap_{\H \in H_{\P}^k} \H$ and $C_k = \bigcup_{\P \in P} C_{\P,k}$. By construction, forall $k \in \N$,
\[
B(A^kz, \varepsilon_k) \cap \I = B(A^kz, \varepsilon_k) \cap \pi_{\last}^{-1}(C_k).
\]
We make three claims.
\begin{itemize}
\item $C_k$ has full dimension (over $\R$) $2s$. Indeed, since $z$ avoids $\bigcup_{\H \in H_{general}} \partial \H$, and $\dim_\R(\P_0 \cap Q)=2s$, $C_0$ has full dimension. Since $\diag\{\lambda_1,\dots,\lambda_s\} C_k \subseteq C_{k+1}$, the claim follows by a easy induction. 
\item $C_k$ is not all of $\C^s$. For this, let us consider $\closure{\compl \I},$ a closed semilinear set which is stable for $A^{-1}$. Under an appropriate diagonal change of basis, which, in particular, stabilizes any set of of form $\prod_{J \in \J} \C^{p_J} \times \{0\}^{d(J) - p_J}$, $A^{-1}$ rewrites as $\diag(\Jord{d(J)}{\lambda_J^{-1}}, J \in \J).$ Hence Lemma~\ref{lem:inter} applies to $\closure{\compl \I}$. Since $\I \neq \C^s$, $\compl \I$ is nonempty, and since it is an open set, it must be fully dimensional. Hence, either $\closure{\compl \I} \subseteq \F$, that is, each point of $Q$ (in particular, $A^kz$) has arbitrary close points that are not in $\I$, which implies the claim.
\item There are finitely many different sets $C_k$ for $k$ in $\N$. Indeed, $C_k$ is determined by finitely many queries, namely whether $A^k z$ is in $\interior{\H}, \partial \H$ or not in $\H$, for each $\H$ in $H$. Note that on the other hand, $\varepsilon_k$ does depend on $k$, and may take arbitrarily small values if $A^k z$ gets arbitrarily close to some $\H$ in $H$ when $k$ ranges in $\N$.
\end{itemize}
As previously stated, $\{(\lambda_1^k,\dots,\lambda_s^k), k \in \N\}$ is dense in $\{(\lambda_1^t,\dots,\lambda_s^t), t \in \R\} \subseteq T_A$. Hence, there exists an increasing sequence $\varphi : \N \to \N$ and $\varepsilon_k/2 \leq \mu_k \leq \varepsilon_k$ such that forall $k$, $(\lambda_1^{\varphi(k)},\dots,\lambda_s^{\varphi(k)}) = (\lambda_1^{\mu_k},\dots,\lambda_s^{\mu_k})$. Let $C$ be such that $C=C_{\varphi(k)}$ for infinitely many $k$.
Combining the diagonal case from section~\ref{subsec:diag}, the fact that $C$ has full dimension, and that $C$ is not $\C^s$, we know that $C$ cannot stabilize $\diag(\lambda_1,\dots,\lambda_s)$. In particular, there is $\tilde u\in C$ such that $\diag(\lambda_1,\dots,\lambda_s) \tilde u \notin C$. Let $t_0 = \inf\{t \leq 1, \diag(\lambda_1^t,\dots,\lambda_s^t)\tilde u \notin C\} \geq 0$, and $u = \diag(\lambda_1^{t_0},\dots,\lambda_s^{t_0})\tilde u \in C$, since $C$ is closed. Note that for any small enough $\varepsilon>0, \diag(\lambda_1^{\varepsilon},\dots,\lambda_s^{\varepsilon}) u \notin C$. We let $N \in \N$ be such that for $n \geq N$, $\varepsilon_n$ is small enough in this sense, and $N'$ be such that $\varphi(N')-\varphi(0) \geq N$. Recall that $C$ is defined using half-spaces which contain $0$ in their border, hence it is invariant under multiplication by positive reals (a cone). Hence we may assume that $||u|| \leq 2^{-\varphi(N')}\varepsilon_{\varphi(N')}$.

We may finaly give the last construction. Let $v \in B(A^{\varphi(0)}z, 2^{-\varphi(N')}\varepsilon_{\varphi(N')}) \cap \pi_{last}^{-1}(\{u\}) \subseteq B(A^{\varphi(0)}z, \varepsilon_{\varphi(0)}) \cap \pi_\last^{-1}(C) \subseteq \I$. We argue that $A^{\varphi(N') - \varphi(0)}v \in B(A^{\varphi(N')}z, \varepsilon_{\varphi(N')})$. Indeed, $A$ is $2$-lipschitzian, so $A^{\varphi(N') - \varphi(0)}$ is $2^{\varphi(N')}$-lipschitzian, so
\[
||A^{\varphi(N') - \varphi(0)}v - A^{\varphi(N')}z || \leq 2^{\varphi(N')} ||v-A^{\varphi(0)}z|| \leq \varepsilon_{\varphi(N')}.
\]
Hence, $A^{\varphi(N') - \varphi(0)}v \in B(A^{\varphi(N')}z, \varepsilon_{\varphi(N')}) \cap \I = B(A^{\varphi(N')}z, \varepsilon_{\varphi(N')}) \cap \pi_\last^{-1}(C)$, but $\pi_\last(A^{\varphi(N') - \varphi(0)}v) = \diag (\lambda_1^{\varphi(N')},\dots,\lambda_s^{\varphi(N')})u = \diag(\lambda_1^{\mu_{N'}},...,\lambda_s^{\mu_{N'}})$, and since $\mu_{N'} \leq \varepsilon_{N'} \leq \varepsilon_{N},$ we obtain that $\pi_\last(A^{\varphi(N') - \varphi(0)}v) \notin C,$ a contradiction.
\end{proof}

Note that Theorem~\ref{thm:all_mod_1_not_equiv} gives a minimal semilinear invariant for $A$ which contains a given $x \in \C^d$, namely,
\[
\I = \prod_{J \in \J} \C^{p_J} \times \{0\}^{d(J) - p_J},
\]
where $p_J = \max\{i \leq d(J) \mid x_{J,i} \neq 0\}$. In particular, if the instance is reduced, $x_{J,d(J)} \neq 0$ so the minimal seminlinear invariant is $\C^d$.

\subsection{Some eigenvalues may be equivalent}\label{subsec:mod_one_not_rou}

We will show that this case reduces to the previous one. We first deal with the case where all equivalent eigenvalues are in fact equal, which we then extend to the general case. Let us start with a Lemma.

\begin{lemma}\label{lem:reducdim}
Let
\[
A= \begin{bmatrix}
A'' \\
& \Jord{d_1}{\lambda} \\
& & \Jord{d_2}{\lambda}
\end{bmatrix},
\]
and $x = (x'', x_1, x_2) \in \C^s$ with $s=s'' + d_1 + d_2$, $s''$ being the dimension of $A''$ (hence $s$ is the dimension of $A$). The $i-th$ coordinate of $x$ in the $\Jord{d_1} \lambda$ block (resp. in the $\Jord{d_2} \lambda$ block) is denoted $x_{1,i}$ (resp. $x_{2,i}$). We assume that $x_{1, d_1} \neq 0$ and let 
\[
A' = \begin{bmatrix}
A'' \\
 & J_{d_1}(\lambda) \\
 & \frac {x_{2,d_2}}{x_{1,d_1}} E_{d_2-1,d_1} & J_{d_2 - 1}(\lambda)
\end{bmatrix},
\]
of size $s'=s-1,$ where $E_{d_2-1,d_1}$ denotes the matrix with $d_2-1$ rows and $d_1$ columns with a single 1 in the bottom right corner. We let $x' = (x'', x_1, x'_2)$ be given by $x'_2 = (x_{2,1}, x_{2,2} \dots, x_{2, d_2-1})$. We assume that $\I' \subseteq \C^{s'}$ is a minimal semilinear invariant for $A'$ and $x'$. Then
\[
\I = \{z \in \C^s \mid z' \in \I' \text{ and } x_{1,d_1} z_{2,d_2} = x_{2, d_2} z_{1,d_1}\} \subseteq \C^s
\]
is a minimal semilinear invariant for $x,A$.
\end{lemma}

Before going on to the proof, let us remark that if $d_2=1$, $E_{d_2-1, d_1}$ and $J_{d_2 - 1}(\lambda)$ are both empty matrices (and the Lemma also holds).

\begin{proof}
Clearly $x \in \I$. Let us first check that $\I$ is invariant for $A$. Let $z \in \I$. Then

\[
\begin{aligned}
(Az)' &= \left(A'' z'', J_{d_1}(\lambda) z_1, \left(J_{d_2}(\lambda)z \right)' \right) \\
&= \left(A'' z'', J_{d_1}(\lambda) z_1, \left(J_{d_2}(\lambda)z \right)_1, \dots,\left(J_{d_2}(\lambda)z \right)_{d_2 - 1}  \right) \\
&= \left(A'' z'', J_{d_1}(\lambda) z_1, \left(J_{d_2}(\lambda)z \right)_1, \dots,\left(J_{d_2}(\lambda)z \right)_{d_2 - 2}, \lambda z_{2, d_2 -1} + z_{2, d_2} \right) \\
&= \left(A'' z'', J_{d_1}(\lambda) z_1, \left(J_{d_2}(\lambda)z \right)_1, \dots,\left(J_{d_2}(\lambda)z \right)_{d_2 - 2}, \lambda z_{2, d_2 -1} + \frac{x_{2, d_2}}{x_{1, d_1}}z_{1, d_1} \right) \\
&= A'z' \in \I',
\end{aligned}
\]
and $x_{1, d_1} (Az)_{2, d_2} = x_{1, d_1} \lambda z_{2, d_2} = x_{2, d_2} \lambda z_{1, d_1} = x_{2, d_2} (Az)_{1, d_1}$. Hence $\I$ is invariant for $A$.

We now show minimality. Let $\P$ be a semilinear invariant containing $x$. Consider $ \P_0= \P \cap \{z \mid x_{1, d_1} z_{2, d_2} = x_{2, d_2} z_{1, d_1}\}$, and $\P_0' \subseteq \C^{s'}$ be its projection on all but the last coordinate. We show that $\P_0'$ is invariant for $A'$. Let $z' \in \P_0'$, and let $z$ be $z'$ extended with $z_{2, d_2} = \frac{x_{2, d_2}}{x_{1, d_1}} z_{1, d_1}$. Then $z \in \P_0$, so $Az \in P_0$. Now, $A'z' = (Az)'$ (the proof for this is similar as that of $\I$'s stability), and so $A'z' \in P_0'$. Hence, $\I' \subseteq P_0'$ by minimality of $\I'$, and so $\I \subseteq \P_0 \subseteq P$. 
\end{proof}

Through repeated applications of Lemma~\ref{lem:reducdim} which reduce the dimension and Lemma~\ref{lem:reduc} which renormalize the reduced instance, we obtain the following Theorem.

\begin{theorem}\label{thm:equalornonequiv}
Let $A$ have its eigenvalues $\lambda_1,\dots, \lambda_s \notin \U$ such that either $\lambda_i = \lambda_j$ or $\lambda_i / \lambda_j \notin \U$, and $x \in \C^s$ which is nonzero on the last coordinate of each block. Then there exists a minimal semilinear invariant for $A,x$ which can be constructed in polynomial time and has polynomial size.
\end{theorem}

\begin{proof}
If forall $i \neq j$, $\lambda_i \neq \lambda_j$, then they are pairwise non-equivalent, so by Theorem~\ref{thm:all_mod_1_not_equiv}, $\C^s$ is the minimal invariant for $A$ and $x$. Otherwise, we use Lemma~\ref{lem:reducdim} on blocks with equal eigenvalues to reduce the dimension, and Lemma~\ref{lem:reduc} to normalize the instance (ensuring that $A$ is in Jordan normal form and $x$ is nonzero on the last coordinate of each bloc), and easily conclude by induction.
\end{proof}

We may now finally extend to the general case.

\begin{theorem}
Let $A$ have only eigenvalues of modulus 1 which are not roots of unity, and $x \in \C^s$ with nonzero last coordinates. Then there is an explicit minimal semilinear invariant $\I$ for $A$ and $x$. In particular, there is a semilinear invariant (namely, $\I$) for $\ell=(A,x,y)$ if and only if $y \notin \I$, which may be decided algorithmically.
\end{theorem}

\begin{proof}
Let $N \in \N$ be such that $A^N$ is just like in the statement of Theorem~\ref{thm:equalornonequiv}. Let $\I_0$ be the minimal semilinear invariant for $A^N$ and $x$. Let
\[
\I = \bigcup_{i=0}^{N-1} A^i \I_0.
\]
Clearly, $\I$ is semilinear, contains $x$, and invariant for $A$. Let $\P$ be a semilinear invariant for $A$ which contains $x$. Then it is also invariant for $A^N$, so $\I_0 \subseteq \P$. It follows that $A \I_0, A^2 \I_0, \dots \subseteq \P$, which concludes our proof. 
\end{proof}

\end{document}